\PassOptionsToPackage{square, comma, numbers, sort&compress}{natbib}
\documentclass[final, 5p, times, twocolumn]{elsarticle}

\usepackage[OT1]{fontenc}
\usepackage{cite, color, bm, multirow, svg}
\usepackage{paralist, graphicx, enumerate, amssymb}
\usepackage{array, dsfont, subfigure, mathrsfs, epstopdf}
\usepackage{amssymb, stfloats}
\usepackage[cmex10]{amsmath}
\usepackage[mathscr]{eucal}

\newtheorem{theorem}{Theorem}

\newtheorem{proposition}{Proposition}
\newtheorem{definition}{Definition}
\newtheorem{assumption}{Assumption}

\newtheorem{example}{Example}
\newtheorem{claim}{Claim}
\newdefinition{remark}{Remark}
\newproof{proof}{Proof}

\DeclareMathOperator\arctanh{arctanh}
\DeclareMathOperator{\p}{\mathrm{p}}
\DeclareMathOperator{\ct}{\mathrm{c}}
\DeclareMathOperator{\mati}{mati}
\DeclareMathOperator{\miati}{miati}

\DeclareMathOperator{\pr}{\mathbf{P}}

\DeclareMathOperator{\diag}{diag}

\DeclareMathOperator{\dom}{dom}
\DeclareMathOperator{\graph}{graph}

\definecolor{bluencs}{rgb}{0.0, 0.53, 0.74}

\journal{}

\begin{document}

\begin{frontmatter}

\title{Emulation-based Stabilization for Networked Control Systems \\ with Stochastic Channels\tnoteref{label0}}
\tnotetext[label0]{This work was supported by the Fundamental Research Funds for the Central Universities under Grant DUT22RT(3)090, the National Natural Science Foundation of China under Grants 61890920, 61890921 and 08120003, the Australian Research Council under the Discovery Project DP200101303, and AFOSR grant FA9550-21-1-0452.}

\author[label1,label2]{Wei Ren\corref{cor1}}\ead{wei.ren@dlut.edu.cn}
\author[label3]{Wei Wang}\ead{wei.wang@unimelb.edu.au}
\author[label1,label2]{Zhuo-Rui Pan}\ead{panzhuorui@mail.dlut.edu.cn}
\author[label1,label2]{Xi-Ming Sun}\ead{sunxm@dlut.edu.cn}
\author[label4]{Andrew R. Teel}\ead{teel@ece.ucsb.edu}
\author[label3]{Dragan Ne{\v{s}}i{\'c}}\ead{dnesic@unimelb.edu.au}

\cortext[cor1]{Corresponding author.}
\address[label1]{Key Laboratory of Intelligent Control and Optimization for Industrial Equipment of Ministry of Education, Dalian University of Technology, Dalian 116024, China.}
\address[label2]{School of Control Science and Engineering, Dalian University of Technology, Dalian 116024, China.}
\address[label3]{Department of Electrical and Electronic Engineering, University of Melbourne, Parkville, 3052, Victoria, Australia.}
\address[label4]{Department of Electrical and Computer Engineering, University of California Santa Barbara, Santa Barbara, CA 93106, USA.}

\begin{abstract}
This paper studies the stabilization problem of networked control systems (NCSs) with random packet dropouts caused by stochastic channels. To describe the effects of stochastic channels on the information transmission, the transmission times are assumed to be deterministic, whereas the packet transmission is assumed to be random. We first propose a stochastic scheduling protocol to model random packet dropouts, and address the properties of the proposed stochastic scheduling protocol. The proposed scheduling protocol provides a unified modelling framework for a general class of random packet dropouts due to different stochastic channels. Next, the proposed scheduling protocol is embedded into the closed-loop system, which leads to a stochastic hybrid model for NCSs with random packet dropouts. Based on this stochastic hybrid model, we follow the emulation approach to establish sufficient conditions to guarantee uniform global asymptotical stability in probability. In particular, an upper bound on the maximally allowable transmission interval is derived explicitly for all stochastic protocols satisfying Lyapunov conditions that guarantee uniform global asymptotic stability in probability. Finally, two numerical examples are presented to demonstrate the derived results.
\end{abstract}


\end{frontmatter}

\section{Introduction}
\label{sec-introduction}

Networked control systems (NCSs) are dynamical systems in which sensors, controllers and/or actuators of the plant are physically distributed and communicate via digital channels \citep{Gupta2010networked}. The network is typically modeled as a serial communication channel with a finite bandwidth. NCSs have received considerable attention over the past decades due to numerous advantages over conventional control systems in terms of flexibility, maintenance and cost \citep{Donkers2011stability, Nesic2009unified, Heemels2010networked}. However, the communication over digital channels inevitably results in some issues, such as time-varying transmission intervals, packet dropouts, time-varying delays, scheduling protocols and quantization, which complicate the analysis and design on the one hand and deteriorate system performance and even result in instability on the other hand. As a result, system modeling, stability analysis and controller design of NCSs under these network-induced issues have attracted considerable attention in the past few decades. Many existing results can be found in \citep{Donkers2011stability, Carnevale2007lyapunov, Heemels2010networked, Nesic2004input, Nesic2009unified, Antunes2013stability} and references therein dealing with some of these issues.

For NCSs with the above issues, there are generally two main modelling strategies in the literature: the deterministic modelling strategy \citep{Nesic2004input, Hu2003stochastic, Donkers2011stability, Nesic2009unified} and the stochastic modelling strategy \citep{Hespanha2006stochastic, Donkers2012stability, Antunes2013stability, Liu2021stabilization, Hu2017modeling, Liu2021event}. In the deterministic modelling strategy, the effects of different issues are assumed to be constrained deterministically or transformed into some deterministic constraints. For instance, both transmission intervals and transmission delays are bounded in some intervals \citep{Carnevale2007lyapunov, Heemels2010networked}, and the emulation-based approach is applied to investigate the upper bounds of these intervals to ensure system performance. In \citep{Heemels2010networked, Van2012discrete}, the effects of packet dropouts on the network transmission are transformed into a tighter constraint on the upper bound of the transmission intervals. These deterministic constraints can be embedded into deterministic models to facilitate the modelling and analysis of NCSs. However, these deterministic models may result in conservatism in stability conditions, and are not available for stochastic channels \citep{Park2009generalized, Liu2021stabilization, Li2021output, Wang2007robust, Schenato2007foundations}. For instance, due to the transformation of the effects of packet dropouts on the network transmission \citep{Heemels2010networked, Van2012discrete}, the upper bound of the transmission intervals needs to be small enough to ensure the system stability \citep{Maass2019lp}. Stochastic channels result in random packet dropouts, which cannot be included into deterministic models.

To characterize the stochastic nature of network-induced issues more accurately and to derive less conservative results, stochastic models have been proposed for NCSs with different issues. Among these issues, intermittent data packet dropouts are attributes of non-ideal or stochastic channels and are a main reason for the performance deterioration. For instance, random packet dropouts occur due to a busy channel in carrier sense multiple access (CSMA) schemes \citep{Tabbara2008input, Park2009generalized}, packet collisions in WirelessHART networks \citep{Maass2019stabilization} and the acknowledge timeout in TCP-like networks \citep{Schenato2007foundations}. For random packet dropouts, some stochastic models have been proposed, such as Bernoulli distributed sequences \citep{Wang2007robust, Liu2021stabilization, Liu2021event, Hespanha2006stochastic} and Markovian models \citep{Tabbara2008input, Quevedo2012robust}. These models are either only applicable to some simple cases or based on some unrealistic assumptions. For instance, only packet dropouts are considered in \citep{Wang2007robust, Liu2021stabilization, Quevedo2012robust}, only stochastic transmission intervals are addressed in \citep{Antunes2013stability}, and the sequence of transmission intervals is assumed to be independent and identically distributed in \citep{Tabbara2008input, Maass2019stabilization}. However, the transmission intervals are bounded deterministically while random packet dropouts occur due to the nature of many practical networks like IEEE 802.11 MAC protocol \citep{Cali1998ieee} and IEEE 802.15.4 network \citep{Park2009generalized}. Therefore, it is necessary to explore the modelling and analysis of NCSs with random packet dropouts.

Motivated by these discussions, in this paper we consider NCSs with stochastic channels, and investigate the system modelling and stability analysis problems based on stochastic hybrid system theory. First, we propose a general framework to model the effects of random packet dropouts caused by stochastic channels. More specifically, the transmission intervals are assumed to be deterministic whereas the information transmission is random, which are characteristics that are different from those in \citep{Antunes2013stability, Tabbara2008input, Maass2019stabilization} as this is consistent with many practical networks \citep{Cali1998ieee, Park2009generalized}. In this respect, a unified stochastic difference system is developed to model random packet dropouts. Second, the stochastic difference model is embedded into the closed-loop system, thereby resulting in a novel stochastic hybrid model for NCSs with random packet dropouts. With the stochastic hybrid model, the emulation-based approach \citep{Wang2014emulated, Carnevale2007lyapunov} is implemented to propose a novel bound for the maximal allowable transmission interval (MATI) to guarantee the stochastic stability property of the whole system. Finally, we apply the derived results to the case of linear systems with specific networks, and show the flexibility and generality of the proposed model and approach. It has to be noticed that we focus on Ethernet-like and WirelessHART networks \citep{Tabbara2008input, Maass2019stabilization}, and that we propose a different stability analysis as well as a different MATI bound to show the guarantee of the system stability under the stochastic channel.  

The rest of this paper is organized as follows. The notation and preliminaries on stochastic hybrid systems are introduced in Section \ref{sec-preliminaries}. The stochastic protocol and the overall model are proposed in Section \ref{sec-emulation}. The properties of stochastic protocols are analyzed in Section \ref{sec-stochasticprotocol}. The main results are derived in Section \ref{sec-stability}. The application to linear systems with specific networks is given in Section \ref{sec-casestudy}. Two numerical examples are presented in Section \ref{sec-example}. Conclusion and future work are given in Section \ref{sec-conclusion}.

\section{Preliminaries}
\label{sec-preliminaries}

$\mathbb{R}:=(-\infty, +\infty); \mathbb{R}_{+}:=[0, +\infty); \mathbb{N}:=\{0, 1, \ldots\}$ and $\mathbb{N}_{+}:=\{1, 2, \ldots\}$. Given $x\in\mathbb{R}^{n}, y\in\mathbb{R}^{m}$, $(x, y):=(x^{\top}, y^{\top})^{\top}$ for simplicity. Given $\mathcal{A}, \mathcal{B}\subset\mathbb{R}^{n}$, $\mathcal{A}\setminus\mathcal{B}:=\{x\in\mathbb{R}^{n}: x\in\mathcal{A}, x\notin\mathcal{B}\}$ and $\mathcal{A}+\mathcal{B}:=\{x+y\in\mathbb{R}^{n}: x\in\mathcal{A}, y\in\mathcal{B}\}$. $I$ denotes the identity matrix. $|\cdot|$ denotes the Euclidean norm. For a closed set $\mathcal{A}\subset\mathbb{R}^{n}$ and $x\in\mathbb{R}^{n}$, $|x|_{\mathcal{A}}:=\inf_{y\in\mathcal{A}}|x-y|$ is the Euclidean distance of $x$ to $\mathcal{A}$. $\mathbb{B}$ and $\mathbb{B}^{\circ}$ denote respectively the closed and open unit ball in $\mathbb{R}^{n}$. $\pr\{\cdot\}$ denotes a probability measure. A set-valued mapping $\mathcal{M}: \mathbb{R}^{p}\rightrightarrows\mathbb{R}^{n}$ is \emph{outer semicontinuous} if, for each $(x_{i}, y_{i})\rightarrow(x, y)\in\mathbb{R}^{p}\times\mathbb{R}^{n}$ satisfying $y_{i}\in\mathcal{M}(x_{i})$ for all $i\in\mathbb{N}_{+}$, $y\in\mathcal{M}(x)$. A mapping $\mathcal{M}$ is \emph{locally bounded} if for each bounded set $\mathbb{K}\subset\mathbb{R}^{p}$, $\mathcal{M}(\mathbb{K}):=\cup_{x\in\mathbb{K}}\mathcal{M}(x)$ is bounded. Given a function $f: \mathbb{R}_{+}\rightarrow\mathbb{R}^{n}$, $\|f\|$ denotes the supremum norm on $[0, \infty)$; $f(t^{+}):=\limsup_{s\rightarrow0^{+}}f(t+s)$, and $f^{+}:=f(t^{+})$. Given a function $V: \mathbb{R}^{n}\rightarrow\mathbb{R}_{+}$ that is locally Lipschitz on an open set $\mathcal{O}\subset\mathbb{R}^{n}$, a point $x\in\mathcal{O}$ and a vector $f\in\mathbb{R}^{n}$, we denote by $V^{\circ}(x; f)$ the Clarke generalized directional derivative of $V$ at $x$ in the direction $f$. If $V$ is continuously differentiable, then $V^{\circ}(x; f)=\langle\nabla V(x), f\rangle$. Given a closed set $\mathcal{A}\subset\mathbb{R}^{n}$, $\rho: \mathbb{R}^{n}\rightarrow\mathbb{R}_{+}$ is of class $\mathcal{PD}(\mathcal{A})$, if it is continuous and for each $\delta>0$ and $\Delta>0$ there exists $\bar{\rho}>0$ such that $\rho(x)\geq\bar{\rho}$ for all $x\in(\mathcal{A}+\Delta\mathbb{B})\setminus(\mathcal{A}+\delta\mathbb{B}^{\circ})$. A function $\alpha: \mathbb{R}_{+}\to\mathbb{R}_{+}$ is of class $\mathcal{K}$, if it is continuous, zero at zero, and strictly increasing; $\alpha$ is of class $\mathcal{K}_{\infty}$ if it is of class $\mathcal{K}$ and unbounded. A function $\beta: \mathbb{R}_{+}\times\mathbb{R}_{+}\to\mathbb{R}_{+}$ is of class $\mathcal{KL}$, if $\beta(s, t)$ is of class $\mathcal{K}$ for fixed $t\geq0$ and decreases to zero as $t\rightarrow\infty$ for fixed $s\geq0$.

In this paper, we will model NCSs with random packet dropouts as the following stochastic hybrid system \citep{Teel2014stability, Subbaraman2017robust}:
\begin{subequations}
\label{eqn-1}
\begin{align}
\label{eqn-1a}
\dot{\xi}&\in\mathcal{F}(\xi), \quad \xi\in C, \\
\label{eqn-1b}
\xi^{+}&\in\mathcal{G}(\xi, \upsilon^{+}),  \quad \xi\in D, \\
\label{eqn-1c}
\upsilon&\sim\mu(\cdot),
\end{align}
\end{subequations}
where $\xi\in\mathbb{R}^{n}$ is the state, $\upsilon\in\mathbb{R}^{p}$ is the random input, and $C, D\subset\mathbb{R}^{n}$ are the flow and jump sets, respectively. $\mathcal{F}: \mathbb{R}^{n}\rightrightarrows\mathbb{R}^{n}$ and $\mathcal{G}: \mathbb{R}^{n}\times\mathbb{R}^{p}\rightrightarrows\mathbb{R}^{n}$ are the flow and jump maps, respectively. The distribution function $\mu$ is derived from the probability space $(\Omega, \mathfrak{F}, \pr)$ and a sequence of independent and identically distributed (i.i.d.) input variables $\mathbf{v}_{i}: \Omega\rightarrow\mathbb{R}^{p}$ defined on $(\Omega, \mathfrak{F}, \pr)$ for $i\in\mathbb{N}_{+}$. That is, $\mu$ is defined as $\mu(\mathcal{A}):=\pr\{\omega\in\Omega: \mathbf{v}_{i}(\omega)\in\mathcal{A}\}$ for any $\mathcal{A}\in\mathbf{B}(\mathbb{R}^{p})$. Let $\mathfrak{F}_{i}$ be the natural filtration of $\mathfrak{F}$ with respect to the random variables $\{\mathbf{v}_i\}_{i=1}^{\infty}$. That is, $\mathfrak{F}_{i}:=\sigma\{\mathbf{v}^{-1}_{j}(\mathcal{A}): 1\leq j\leq i, \mathcal{A}\in\mathfrak{F}\}$, which means the smallest $\sigma$-algebra on $(\Omega, \mathfrak{F})$ that contains the pre-images of $\mathbf{B}(\mathbb{R}^{p})$-measurable subsets on $\mathbb{R}^{p}$ for times up to $i$. It was shown in \citep{Subbaraman2017robust} that, if the following assumption holds, then the system \eqref{eqn-1} is well-posed to guarantee the  existence of nontrivial solutions \citep{Teel2013lyapunov}.

\begin{assumption}
\label{asp-1}
For the system \eqref{eqn-1}, the following holds.
\begin{enumerate}[(i)]
  \item The sets $C, D\subset\mathbb{R}^{n}$ are closed.\vspace{-5pt}
  \item The mapping $\mathcal{F}: \mathbb{R}^{n}\rightrightarrows\mathbb{R}^{n}$ is outer-semicontinuous, locally bounded with nonempty convex values on $C$.\vspace{-5pt}
  \item The mapping $\mathcal{G}: \mathbb{R}^{n}\times\mathbb{R}^{n}\rightrightarrows\mathbb{R}^{n}$ is locally bounded and the mapping $\upsilon\mapsto\graph(\mathcal{G}(\cdot, \upsilon)):=\{(x, y)\in\mathbb{R}^{2n}: y\in\mathcal{G}(x, \upsilon)\}$ is measurable with closed values.
\end{enumerate}
\end{assumption}

The notion of a solution to non-stochastic hybrid systems studied in \citep{Goebel2012hybrid} is introduced below. A \emph{compact hybrid time domain} is of the form $\cup^{J}_{j=0}([t_{j}, t_{j+1}]\times\{j\})\subset\mathbb{R}_{+}\times\mathbb{N}_{+}$ for some $J\in\mathbb{N}_{+}$ and real numbers $0=t_{0}\leq t_{1}\ldots\leq t_{J+1}$. A set $E\subset\mathbb{R}_{+}\times\mathbb{N}_{+}$ is a \emph{hybrid time domain}, if it is a union of a finite or infinite sequence of intervals $[t_{j}, t_{j+1}]\times\{j\}$, with the last interval (if existent) possibly of the form $[t_{j}, T)$ with $T$ finite or $T=\infty$. That is, a hybrid time domain is the union of a nondecreasing sequence of compact hybrid time domains. A mapping $\phi: E\rightarrow\mathbb{R}^{n}$ is a \emph{hybrid arc}, if $E$ is a hybrid time domain and for each $j\in\mathbb{N}_{+}$, the mapping $t\mapsto\phi(t, j)$ is locally absolutely continuous on $\{t: (t, j)\in E\}$.

\emph{A stochastic hybrid arc} is a mapping $\xi$ defined on $\Omega$ such that $\xi(\omega)$ is a hybrid arc for each $\omega\in\Omega$ and the set-valued mapping from $\Omega$ to $\mathbb{R}^{n+2}$ defined by $\omega\rightarrow\graph(\xi(\omega)):=\{(t, j, z): \phi=\xi(\omega), (t, j)\in\dom(\phi), z=\phi(t, j)\}$ is $\mathfrak{F}$-measurable with closed values. Define $\graph(\xi(\omega))_{\leq j}:=\graph(\xi(\omega))\cap(\mathbb{R}_{+}\times\{0, \ldots, j\}\times\mathbb{R}^{n})$. An $\{\mathfrak{F}_{j}\}^{\infty}_{j=0}$ adapted stochastic hybrid arc is a stochastic hybrid arc $\xi$ such that the mapping $\omega\rightarrow\graph(\xi(\omega))_{\leq j}$ is $\mathfrak{F}_{j}$-measurable for each $j\in\mathbb{N}$. \emph{An adapted stochastic hybrid arc} $\xi$ is a solution starting from $\zeta$, denoted as $\xi\in\mathcal{S}(\zeta)$, if $\xi(\omega)$ is a solution to \eqref{eqn-1} with inputs ${\mathbf{v}_{i}(\omega)}^{\infty}_{i=0}$, that is, with $\phi_{\omega}:=\xi(\omega)$,
\begin{enumerate}[1)]
  \item $\phi_{\omega}(0, 0)=\zeta$;\vspace{-5pt}
  
  \item for $(t_{1}, j), (t_{2}, j)\in\dom(\phi_{\omega})$ with $t_{1}<t_{2}$, we have $\phi_{\omega}(t, j)\in C$ and $\dot{\phi}_{\omega}(t, j)\in\mathcal{F}(\phi_{\omega}(t, j))$ for almost every $t\in[t_{1}, t_{2}]$; \vspace{-5pt}
      
  \item for $(t, j), (t, j+1)\in\dom(\phi_{\omega})$, we have $\phi_{\omega}(t, j)\in D$ and $\phi_{\omega}(t, j+1)\in\mathcal{G}(\phi_{\omega}(t, j), \mathbf{v}_{j+1}(\omega))$.
\end{enumerate}

A random solution is (pathwise) \emph{maximal} if for each sample path the domain of the sample path cannot be extended further; see \citep[Definition 2.7]{Goebel2012hybrid} for more details. A random solution is \emph{almost surely complete} if for almost every $\omega$, the sample path $\xi(\omega)$ is complete. A sample path is said to be \emph{complete} if the domain of the sample path is unbounded.

\begin{definition}[{\citep{Subbaraman2016recurrence}}]
\label{def-1}
A compact set $\mathcal{A}\subset\mathbb{R}^{n}$ is \emph{uniformly globally stable in probability (UGSp)} for the system \eqref{eqn-1}, if
\begin{enumerate}
  \item for each $\varepsilon>0$ and $\rho>0$, there exists $\delta>0$ such that $\pr\{\graph(\xi)\subset\mathbb{R}^{2}\times(\mathcal{A}+\varepsilon\mathbb{B}^{\circ})\}\geq1-\rho$ for all $\xi\in\mathcal{S}(\mathcal{A}+\delta\mathbb{B})$;\vspace{-5pt}
  \item for each $\delta>0$ and $\rho>0$, there exists $\varepsilon>0$ such that $\pr\{\graph(\xi)\subset\mathbb{R}^{2}\times(\mathcal{A}+\varepsilon\mathbb{B}^{\circ})\}\geq1-\rho$  for all $\xi\in\mathcal{S}(\mathcal{A}+\delta\mathbb{B})$.
\end{enumerate}
The set $\mathcal{A}$ is \emph{uniformly globally attractive in probability (UGAp)} for \eqref{eqn-1} if, for each $\varepsilon>0, \rho>0$ and $\delta>0$, there exists $\varrho>0$ such that $\pr\{(\graph(\xi)\cap(\Gamma_{\geq\varrho}\times\mathbb{R}^{n}))\subset\mathbb{R}^{2}\times(\mathcal{A}+\varepsilon\mathbb{B}^{\circ})\}\geq1-\rho$ for $\xi\in\mathcal{S}(\mathcal{A}+\delta\mathbb{B})$. The set $\mathcal{A}$ is \emph{uniformly globally asymptotically stable in probability (UGASp)} for \eqref{eqn-1} if, it is both UGSp and UGAp for \eqref{eqn-1}.
\end{definition}

\section{Stochastic Hybrid Model for NCS}
\label{sec-emulation}

In this section, we introduce a novel NCS setup by proposing a general stochastic scheduling protocol to capture the effects of random packet dropouts, and then embed this protocol within the stochastic framework of the form \eqref{eqn-1}.

\subsection{Networked Control Configuration}

Consider the following continuous-time plant
\begin{equation}
\label{eqn-2}
\dot{x}_{\p}=f_{\p}(x_{\p}, u),\quad y=g_{\p}(x_{\p}),
\end{equation}
where $x_{\p}\in\mathbb{R}^{n_{\p}}$ is the system state, $u\in\mathbb{R}^{n_{u}}$ is the control input, and $y\in\mathbb{R}^{n_{y}}$ is the system output. Following the emulation-based approach \citep{Wang2014emulated, Maass2019lp}, the first step is to consider the network-free case and to design a controller for \eqref{eqn-2} to guarantee the desired system performance. Here the designed controller is assumed to be of the following form
\begin{equation}
\label{eqn-3}
\dot{x}_{\ct}=f_{\ct}(x_{\ct}, y),\quad u=g_{\ct}(x_{\ct}),
\end{equation}
where $x_{\ct}\in\mathbb{R}^{n_{\ct}}$ is the controller state. The control input depends only on the controller state \citep{Nesic2004input, Hespanha2006stochastic, Wang2007robust, Heemels2010networked}, but can be extended to involve the system output; see, e.g., \citep{Donkers2011stability, Antunes2013stability, Liu2021event, Li2021output}. In \eqref{eqn-2}-\eqref{eqn-3}, $f_{\p}$ and $f_{\ct}$ are assumed to be continuous; $g_{\p}$ and $g_{\ct}$ are assumed to be continuously differentiable. The second step is to implement the designed controller \eqref{eqn-3} over the network and to guarantee that the assumed stability of the system \eqref{eqn-2}-\eqref{eqn-3} will be preserved under some reasonable assumptions.

Since the communication between the plant and controller is via the network, $y\in\mathbb{R}^{n_{y}}$ and $u\in\mathbb{R}^{n_{u}}$ are transmitted via the network. To facilitate the following notation and analysis, we denote by $\hat{y}\in\mathbb{R}^{n_{y}}$ and $\hat{u}\in\mathbb{R}^{n_{u}}$ the vectors of most recently transmitted plant and controller output values via the network. Accordingly, the network-induced error is $e:=(e_{y}, e_{u})\in\mathbb{R}^{n_{e}}$, where $e_{y}:=\hat{y}-y$, $e_{u}:=\hat{u}-u$ and $n_{e}:=n_{y}+n_{u}$.

\subsection{Communication Network and Protocol}
\label{subsec-protocol}

The communication network is limited-capacity and assumed to have $\ell\in\mathbb{N}_{+}$ nodes, which are groups of sensor or actuator signals that are transmitted together in a single packet. Accordingly, the error vector $e$ is partitioned as $e:=(e_{1}, \ldots, e_{\ell})$, where $e_{l}\in\mathbb{R}^{n^{l}_{e}}$ corresponds to the node $l\in\{1, \ldots, \ell\}$ and $\sum^{\ell}_{l=1}n^{l}_{e}=n^{l}_{e}$. At each transmission time $t_{i}\in\mathbb{R}_{+}$, $i\in\mathbb{N}_{+}$, one and only one node is granted access to the network to transmit its data packet. As assumed in \citep{Carnevale2007lyapunov, Nesic2004input1, Nesic2009unified}, all transmission times are required to be strictly increasing and to satisfy
\begin{equation}
\label{eqn-4}
\tau_{\miati}\leq t_{i+1}-t_{i}\leq\tau_{\mati}, \quad \forall i\in\mathbb{N},
\end{equation}
where, $\tau_{\miati}>0$ is called the \emph{minimally achievable transmission interval (MIATI)} and is given by the hardware constraints; $\tau_{\mati}>0$ is called the \textit{maximum allowable transmission interval (MATI)} and is used to measure how fast the network needs to transmit in order to preserve the system performances. Between two successive transmission times, zero-order-hold (ZOH) devices are assumed to be implemented such that $\dot{\hat{y}}=0$ and $\dot{\hat{u}}=0$.

\begin{remark}
One and only one node is granted access to the network at each transmission time, which is a common assumption in many existing works \citep{Walsh2002stability, Nesic2004input1, Antunes2013stability, Heemels2010networked}. This assumption comes from both the limited capacity of communication channels \citep{Nesic2009unified} and many realistic network protocols \citep{Park2009generalized, Park2013modeling}. If multiple nodes can transmit their data packets at each transmission time, then additional difficulties in system modelling and stability analysis need to be resolved, which is out of the scope of this paper and could be a potential topic for the future study. 
\hfill $\square$
\end{remark}

The communication network is generally considered to be reliable, and thus each node that is granted access to the network will transmit its data packet successfully. Such case is called the \emph{dropout-free case}. In this case, the update of the network-induced error is given by (see also \citep{Nesic2004input1})
\begin{align}
\label{eqn-5}
e(t^{+}_{i})=h(i, e(t_{i})),
\end{align}
where the function $h: \mathbb{N}\times\mathbb{R}^{n_{e}}\rightarrow\mathbb{R}^{n_{e}}$ depends only on the network protocol and is independent of the plant and controller dynamics. We usually call \eqref{eqn-5} the \emph{scheduling protocol}. However, the assumption of reliable communication networks does not always hold. Many random phenomena, including the channel access failure, packet collisions and the acknowledge (ACK) timeout \citep{Maass2019stabilization, Tabbara2008input, Park2013modeling}, may occur in the transmission process. Because of these random phenomena, not all data packets are transmitted successfully, thereby resulting in random packet dropouts. This phenomenon is called the \emph{random dropout case}, and cannot be described via the deterministic protocol \eqref{eqn-5}.

To model the random dropout case, we first introduce a time-homogenous binary variable vector $\upsilon_{l}:=(\upsilon_{l1}, \ldots, \upsilon_{lm})\in\{0, 1\}^{m}$ to each node, where $m\in\mathbb{N}_{+}$ and $l\in\{1, \ldots, \ell\}$. This variable vector $\upsilon_{l}$ is used to indicate if the packet dropout occurs when this node is granted to access to the network. Note that the packet dropout may be caused by different reasons, which are independent of each other \citep{Park2009generalized, Park2013modeling}. Hence, $\upsilon_{l}$ is presented into the vector form. Each component in $\upsilon_{l}$ implies a single reason for the packet dropout and is to record if the corresponding reason occurs. That is, if the $l$-th node is chosen, then $\upsilon_{lj}=0$ means the occurrence of the $j$-th reason, where $j\in\{1, \ldots, m\}$; otherwise, the $j$-th reason does not occur. If $\upsilon_{lj}=1$ for all $j\in\{1, \ldots, m\}$, then the $l$-th node transmits its data packet successfully. All these variables are assumed to be independent of each other, and $\sum^{\ell}_{l=1}\sum^{m}_{j=1}\upsilon_{lj}\leq m$. That is, at each transmission time, there exists at most one node being chosen and transmitting its packet successfully. In particular, $\sum^{\ell}_{l=1}\sum^{m}_{j=1}\upsilon_{lj}\neq m$ implies the transmission failure due to the packet dropout caused by at least one reason. Let $\pr\{\upsilon_{lj}=1\}=p_{lj}\in(0, 1)$ be the occurrence probability of the $j$-th reason for the $l$-th node. 
Based on the probabilities of all binary variables, we have the probabilities $P_{\mathrm{s}}, P_{\mathrm{f}}\in(0, 1)$ for the transmission success and failure, respectively. The next example is presented to show the generality and flexibility of the proposed random dropout modelling.

\begin{example}
The following three different cases are presented, which include many existing cases, e.g., in \citep{Tabbara2008input, Maass2019stabilization, Quevedo2012robust, Schenato2007foundations, Wang2007robust}.
\begin{itemize}
  \item For the carrier sense multiple access with collision avoidance (CSMA/CA) mechanism of IEEE 802.15.4 (see \citep{Park2009generalized} for more details), there exist two reasons for packet dropouts: busy channel and packet collision. Hence, $\upsilon_{l}:=(\upsilon_{l1}, \upsilon_{l2})\in\{0, 1\}^{2}$ for each node. Let $\pr\{\upsilon_{l1}=1\}=p_{1}\in(0, 1)$ and $\pr\{\upsilon_{l2}=1\}=p_{2}\in(0, 1)$, which implies that the occurrence probabilities of these two reasons are the same for all nodes. The probabilities of the transmission success and failure at each node are
      \begin{align*}
      P^{l}_{\mathrm{f}}=1-p_{1}p_{2}, \quad P^{l}_{\mathrm{s}}=p_{1}p_{2}.
      \end{align*}
      Hence, for each transmission, the probabilities of transmission failure and success are
      \begin{align*}
      P_{\mathrm{f}}=(1-p_{1}p_{2})(p_{1}p_{2})^{\ell-1}, \quad  P_{\mathrm{s}}=1-P_{\mathrm{f}}.
      \end{align*}\vspace{-20pt}

  \item Let $\upsilon_{l}\in\{0, 1\}$ and the packet dropouts at different transmission times be independent of each other. The probabilities of the transmission success and failure at each node are assumed to be
      \begin{align*}
      \pr\{\upsilon_{l}=1\}=p_{l}\in(0, 1), \quad \pr\{\upsilon_{l}=0\}=1-p_{l}\in(0, 1).
      \end{align*}
      For each transmission, the probabilities of transmission failure and success are
      \begin{align*}
      P_{\mathrm{f}}=\prod_{k\in\{1, \ldots, \ell\}, k\neq l}(1-p_{l})p_{k}, \quad  P_{\mathrm{s}}=1-P_{\mathrm{f}}.
      \end{align*}\vspace{-20pt}

  \item If the packet dropout at the currently-chosen node is related to the case at the previously-chosen node, then the transmission process can be characterized via a time-homogeneous Markov process \citep{Quevedo2012robust}. That is, for each node $l\in\{1, \ldots, \ell\}$, the transition probabilities are given by
      \begin{align}
      \label{eqn-6}
      \begin{aligned}
      \pr\{\upsilon_{l}(t_{i+1})=1|\upsilon_{k}(t_{i})=1\}&=q_{l}, \\
      \pr\{\upsilon_{l}(t_{i+1})=0|\upsilon_{k}(t_{i})=1\}&=1-q_{l}, \\
      \pr\{\upsilon_{l}(t_{i+1})=1|\upsilon_{k}(t_{i})=0\}&=p_{l}, \\
      \pr\{\upsilon_{l}(t_{i+1})=0|\upsilon_{k}(t_{i})=0\}&=1-p_{l},
      \end{aligned}
      \end{align}
      where $k\in\{1, \ldots, \ell\}$ is the chosen node at $t_{i}$, $q_{l}\in(0, 1)$ is the \emph{success rate} to the $l$-th node, and $p_{l}\in(0, 1)$ is the \emph{recovery rate} to the $l$-th node. Since the node to be chosen is not related to the previously-chosen node and may not be uniform, we can see that the rates in \eqref{eqn-6} depend on the node to be chosen instead of the previously-chosen node. In addition, different nodes are allowed to have different probabilities of packet dropouts, which is embedded in the rates of \eqref{eqn-6}. For instance, $q_{l}$ can be written as $\hat{q}_{l}\check{q}_{l}$ with $\hat{q}_{l}\in(0, 1)$ being the probability of choosing the $l$-th node and $\check{q}_{l}\in(0, 1)$ being the probability of transmission success via the $l$-th node. Based on the stationary distribution of this Markov process, the probabilities of transmission failure and success can be computed respectively.
      \hfill $\lhd$
\end{itemize}
\end{example}

Because of random packet dropouts, a novel update of the network-induced error is proposed below:
\begin{align}
\label{eqn-7}
\begin{aligned}
e(t^{+}_{i})&\in\{(I-\mathscr{S}(\upsilon(t^{+}_{i})))e(t_{i})\}+\mathscr{H}(\mathfrak{h}(i, e(t_{i})), \mathscr{S}(\upsilon(t^{+}_{i}))),   \\
\upsilon&\sim\mu(\cdot),
\end{aligned}
\end{align}
where $\upsilon:=(\upsilon_{1}, \ldots, \upsilon_{\ell})\in\{0, 1\}^{m\ell}$, $\mathscr{S}(\upsilon):=\diag(\bar{\upsilon}_{1}I_{1}, \ldots, \bar{\upsilon}_{\ell}I_{\ell})$ with the identity matrix $I_{l}$ corresponding to $e_{l}$, and $\bar{\upsilon}_{l}:=\min_{j\in\{1, \ldots, m\}}\{\upsilon_{lj}\}$. The distribution function $\mu$ is from the probability space $(\Omega, \mathfrak{F}, \pr)$ and the sequence of i.i.d. variables $\mathbf{v}_{i}: \Omega\rightarrow\mathbb{R}^{\ell}$ with $i\in\mathbb{N}_{+}$. More precisely, $\mu(\mathcal{A})=\pr\{\omega\in\Omega: \mathbf{v}_{i}(\omega)\in\mathcal{A}\}$ with $\mathcal{A}\subseteq\{0, 1\}^{m\ell}$. For instance, if $\mathcal{A}=\{0\}^{m\ell}$, then $\mu(\mathcal{A})=\prod_{l\in\{1, \ldots, \ell\}}\prod_{j\in\{1, \ldots, m\}}\pr\{\upsilon_{lj}=0\}$. In \eqref{eqn-7}, the mapping $\mathscr{H}: \mathbb{R}^{n_{e}}\times\{0, 1\}^{m\ell}\rightarrow\mathbb{R}^{n_{e}}$ is assumed to be continuous, while the mapping $\mathfrak{h}: \mathbb{N}\times\mathbb{R}^{n_{e}}\rightrightarrows\mathbb{R}^{n_{e}}$ is assumed to be locally bounded and outer semicontinuous. Here $\mathfrak{h}$ is set-valued so that either possibility for transmission is allowed. In this way, $\mathfrak{h}$ is not required to be continuous. If $\mathfrak{h}$ is reduced to be single-valued, then the assumptions on $\mathfrak{h}$ are equivalent to the continuity requirement; see \citep[Corollary 5.20]{Rockafellar2009variational}. In this case, $\mathfrak{h}$ is similar to but not necessarily the same as the function $h$ in \eqref{eqn-5}. In addition, the mapping $\mathscr{H}$ is to show the coupling between $\mathfrak{h}$ and $\upsilon$. The introduction of $\mathscr{H}$ comes from the nature of some stochastic channels and the potential involvement of the random variable into the complex update mechanism. In Section \ref{sec-casestudy}, we will present some applications to illustrate $\mathfrak{h}$ and $\mathscr{H}$.

Similar to the deterministic case in \citep{Nesic2004input1}, we call \eqref{eqn-7} the \emph{stochastic scheduling protocol}. Compared with the deterministic one \eqref{eqn-5}, the protocol \eqref{eqn-7} shows the effects of random packet dropouts, which correspond to the case $\mathscr{S}(\upsilon)=0$. In this case, 
\begin{equation*}
e(t^{+}_{i})\in\{e(t_{i})\}+\mathscr{H}(\mathfrak{h}(i, e(t_{i})), 0),
\end{equation*}
which is more general than the setting $e(t^{+}_{i})=e(t_{i})$, since packet dropouts may cause other effects on $e$ such that $e(t^{+}_{i})$ does not equal to $e(t_{i})$. The transmission success case is the same as the dropout-free case, and thus $e(t^{+}_{i})=h(i, e(t_{i}))$. In this case, 
\begin{align}
\label{eqn-8}
h(i, e(t_{i}))\in\{(I-S_{l})e(t_{i})\}+\mathscr{H}(\mathfrak{h}(i, e(t_{i})), S_{l}), 
\end{align}
where $l\in\{1, \ldots, \ell\}$ and $S_{l}:=\diag(I_{1}, \ldots, I_{l-1}, 0, I_{l+1}, \ldots, I_{\ell})$ means that the $l$-th node is chosen at the transmission time $t_{i}$ and transmits the data packet successfully. The proposed protocol \eqref{eqn-7} will be further addressed in Section \ref{sec-stochasticprotocol}.

\begin{remark}
The phenomena of packet dropouts have been studied in many works \citep{Schenato2007foundations, Wang2007robust, Maass2019stabilization, Tabbara2008input}, whereas in this work a unified protocol is proposed to show explicitly the effects of packet dropouts on the network transmission. The protocol \eqref{eqn-7} includes these in \citep{Schenato2007foundations, Wang2007robust, Maass2019stabilization, Tabbara2008input} as special cases. In particular, comparing with \citep{Antunes2013stability, Maass2019stabilization, Tabbara2008input} where the sequence of transmission intervals is assumed to be i.i.d., the transmission intervals are bounded in \eqref{eqn-4}, which is a realistic assumption for many network protocols. For instance, the transmission intervals are bounded for IEEE 802.11 MAC protocol (see \citep[Lemma 1]{Cali1998ieee}) and IEEE 802.15.4 network (see \citep{Park2013modeling}).
\hfill $\square$
\end{remark}

\subsection{Stochastic Hybrid Model}
\label{subsec-model}

With the proposed protocol \eqref{eqn-7}, NCS with random packet dropouts can be modeled as follows:
\begin{align}
\label{eqn-9}
\begin{aligned}
\dot{x}&=f(x, e), \quad t\in[t_{i}, t_{i+1}], \\
\dot{e}&=g(x, e), \quad t\in[t_{i}, t_{i+1}],  \\
e(t^{+}_{i})&\in\{(I-\mathscr{S}(\upsilon(t^{+}_{i})))e(t_{i})\}+\mathscr{H}(\mathfrak{h}(i, e(t_{i})), \mathscr{S}(\upsilon(t^{+}_{i}))),   \\
\upsilon&\sim\mu(\cdot),
\end{aligned}
\end{align}
where $x:=(x_{\p}, x_{\ct})\in\mathbb{R}^{n_{\p}+n_{\ct}}$. The vector fields $f, g$ are obtained by direct calculations from \eqref{eqn-2}-\eqref{eqn-3}; see \citep{Nesic2004input1} for more details.

To transform \eqref{eqn-9} into a formal stochastic hybrid model of the form \eqref{eqn-1}, we introduce two variables (see also \citep{Carnevale2007lyapunov}): the clock variable $\tau\in\mathbb{R}_{+}$ to generate jumps in the hybrid model that correspond to data transmissions, and the variable $\kappa\in\mathbb{N}_{+}$ to record the number of transmissions that is needed to implement certain protocols. Let $\xi:=(x, e, \tau, \kappa)\in\mathbb{R}^{n_{x}}\times\mathbb{R}^{n_{e}}\times\mathbb{R}_{+}\times\mathbb{N}$, and \eqref{eqn-9} is rephrased as the following stochastic hybrid system:
\begin{equation}
\label{eqn-10}
\begin{aligned}
\dot{\xi}&=\mathcal{F}(\xi), &\quad& \xi\in C, \\
\xi^{+}&\in\mathcal{G}(\xi, \upsilon^{+}),  &\quad& \xi\in D,  \\
\upsilon&\sim\mu(\cdot), &\quad &
\end{aligned}
\end{equation}
where $\upsilon\in\{0, 1\}^{m\ell}$, and
\begin{equation}
\label{eqn-11}
\begin{aligned}
C&:=\mathbb{R}^{n_{x}}\times\mathbb{R}^{n_{e}}\times[0, \tau_{\mati}]\times\mathbb{N}, \\
D&:=\mathbb{R}^{n_{x}}\times\mathbb{R}^{n_{e}}\times[\tau_{\miati}, \tau_{\mati}]\times\mathbb{N}.
\end{aligned}
\end{equation}
The mappings $\mathcal{F}$ and $\mathcal{G}$ in \eqref{eqn-10} are defined as
\begin{align}
\label{eqn-12}
\mathcal{F}(\xi)&=(f(x, e), g(x, e), 1, 0), \\
\label{eqn-13}
\mathcal{G}(\xi, \upsilon)&=(x, \{(I-\mathscr{S}(\upsilon))e\}+\mathscr{H}(\mathfrak{h}(\kappa, e), \mathscr{S}(\upsilon)), 0, \kappa+1),
\end{align}
where $\mathfrak{h}$ and $\mathscr{H}$ are given in \eqref{eqn-7}. The following proposition shows the well-posedness of the system \eqref{eqn-10}.

\begin{proposition}
\label{prop-1}
The system \eqref{eqn-10} satisfies Assumption \ref{asp-1}.
\end{proposition}

\begin{proof}
From \eqref{eqn-11}, we can see that the sets $C$ and $D$ are closed, which validates the first item of Assumption \ref{asp-1}. Since the functions $f$ and $g$ in \eqref{eqn-12} are continuous, the mapping $\mathcal{F}$ is continuous, which implies from \citep[Corollary 5.20]{Rockafellar2009variational} that the mapping $\mathcal{F}$ is outer semicontinuous and locally bounded. Hence, the second item of Assumption \ref{asp-1} is verified. 
To show the third item, we define $\mathcal{G}$ as $\varnothing$ on $((\mathbb{R}^{n_{x}}\times\mathbb{R}^{n_{e}}\times\mathbb{R}_{+}\times\mathbb{N})\setminus D)\times\mathcal{V}$ with $\mathcal{V}:=\{0, 1\}^{m\ell}$. The power set of $\mathcal{V}$ is denoted by $2^{\mathcal{V}}$. Since $\mathcal{V}$ is finite and closed, $2^{\mathcal{V}}$ is a $\sigma$-algebra and $(\mathcal{V}, 2^{\mathcal{V}})$ is a measurable space. Given an open set $\mathcal{O}\subset(\mathbb{R}^{n_{x}}\times\mathbb{R}^{n_{e}}\times\mathbb{R}_{+}\times\mathbb{N})^{2}$, we have $\{\upsilon\in\mathcal{V}: \graph(\mathcal{G}(\cdot, \upsilon))\cap\mathcal{O}\neq\varnothing\}\in2^{\mathcal{V}}$, which implies that the set $\{\upsilon\in\mathcal{V}: \graph(\mathcal{G}(\cdot, \upsilon))\cap\mathcal{O}\neq\varnothing\}$ is measurable. From \citep[Appendix A.2]{Teel2013lyapunov}, the mapping $\upsilon\rightarrow\graph(\mathcal{G}(\cdot, \upsilon))$ is measurable. Note that the mapping $\mathscr{H}$ is continuous, and from \citep[Corollary 5.20]{Rockafellar2009variational} it is outer semicontinuous and locally bounded. From \citep[Proposition 5.52]{Rockafellar2009variational}, for any fixed $\upsilon\in\mathcal{V}$, $\mathscr{H}(\mathfrak{h}(\kappa, e), \mathscr{S}(\upsilon))$ is outer semicontinuous and locally bounded. Hence, the mapping $\mathcal{G}$ is outer semicontinuous and locally bounded for any fixed $\upsilon\in\mathcal{V}$. That is, the third item of Assumption \ref{asp-1} is valid. As a result, we conclude that the model \eqref{eqn-10} satisfies Assumption \ref{asp-1}.
\hfill $\blacksquare$
\end{proof}

\section{Stochastic Scheduling Protocols}
\label{sec-stochasticprotocol}

In this section the property of the stochastic protocol \eqref{eqn-7} is investigated. The protocol \eqref{eqn-7} shows the mapping relation of the network-induced errors before and after the transmission time $t_{i}\in\mathbb{R}_{+}$. Since the network-induced error may evolve in the continuous-time interval and \eqref{eqn-7} does not show the relation between the network-induced errors at any two successive transmission times, \eqref{eqn-7} cannot be treated as a stochastic discrete-time system. In order to address the property of \eqref{eqn-7}, we introduce the following auxiliary stochastic discrete-time system:
\begin{align}
\label{eqn-14}
\begin{aligned}
e(i+1)&\in\mathbf{h}(i, e(i), \upsilon(i+1)) \\
&=\{(I-\mathscr{S}(\upsilon(i+1)))e(i)\}+\mathscr{H}(\mathfrak{h}(i, e(i)), \mathscr{S}(\upsilon(i+1))),  \\
\upsilon&\sim\mu(\cdot).
\end{aligned}
\end{align}
Note that $i\in\mathbb{N}$ plays the same role as the discrete-time instant as in \citep{Teel2013matrosov} and can be ignored as in \citep{Teel2013equivalent} if no confusion will be caused. Similar to the deterministic case \citep{Nesic2004input1}, we refer to \eqref{eqn-14} simply as a protocol with a slight abuse of terminology.

As stated in Section \ref{subsec-protocol}, the introduction of the mappings $\mathfrak{h}$ and $\mathscr{H}$ is such that the stochastic protocol \eqref{eqn-14} is as general as possible. Due to the term $(I-\mathscr{S}(\upsilon))e$, $\mathfrak{h}$ is not necessarily the same as the function $h$ in \eqref{eqn-5}. Since $\mathfrak{h}$ is not specified, $\mathscr{H}$ is introduced to show the coupling between $\mathfrak{h}$ and the random variable. To show this, two special cases of $\mathscr{H}$ are presented below. For each chosen node, if the update of its network-induced error does not affect the network-induced errors of other nodes, then $\mathfrak{h}$ is reduced to the function $h$ in \eqref{eqn-5} and
\begin{align}
\label{eqn-15}
\mathscr{H}(\mathfrak{h}(i, e), \mathscr{S}(\upsilon))=\mathscr{S}(\upsilon)\mathfrak{h}(i, e).
\end{align}
If the mapping $\mathfrak{h}$ is of the form $Q(i, e)e$ as in \citep{Maass2019stabilization, Tabbara2008input}, where the matrix $Q(i, e)$ is related to $(i, e)$, then
\begin{align}
\label{eqn-16}
\mathscr{H}(\mathfrak{h}(i, e), \mathscr{S}(\upsilon))=\mathbf{F}(Q(i, e), \mathscr{S}(\upsilon))e,
\end{align}
where the mapping $\mathbf{F}: \mathbb{R}^{n_{e}\times n_{e}}\times\{0, 1\}^{m\ell}\rightarrow\mathbb{R}^{n_{e}\times n_{e}}$ shows the coupling between the matrix $Q(i, e)$ and the random variable. Therefore, we can see that the proposed protocol \eqref{eqn-14} is general enough to cover many existing cases. Once $\mathscr{H}$ is determined, the mapping $\mathbf{h}$ in \eqref{eqn-14} is also determined. The case studies of some realistic networks are presented in detail in Section \ref{sec-casestudy}.

\subsection{Lyapunov UGASp Scheduling Protocols}

To investigate the properties of the protocol \eqref{eqn-14}, we first propose the definition of uniform global asymptotic stability in probability, which is central to the subsequent stability analysis.

\begin{definition}
\label{def-2}
Consider the protocol \eqref{eqn-14}. If there exist a function $W: \mathbb{N}_{+}\times\mathbb{R}^{n_{e}}\rightarrow\mathbb{R}_{+}$, $\alpha_{1W}, \alpha_{2W}\in\mathcal{K}_{\infty}$ and $\bar{\rho}\in[0, 1)$ such that for all $(i, e)\in\mathbb{N}_{+}\times\mathbb{R}^{n_{e}}$,
\begin{subequations}
\label{eqn-17}
\begin{align}
\label{eqn-17a}
\alpha_{1W}(|e|)\leq W(i, e)&\leq\alpha_{2W}(|e|),   \\
\label{eqn-17b}
\int_{\upsilon\in\mathcal{V}}W(i+1, \mathbf{h}(i, e(i), \upsilon))\mu(d\upsilon)&\leq\bar{\rho}W(i, e),
\end{align}
\end{subequations}
then the protocol \eqref{eqn-14} is \textit{globally asymptotically stable in probability (UGASp)} with the Lyapunov function $W$.
\end{definition}

\begin{remark}
Definition \ref{def-2} is related to neither the plant \eqref{eqn-2} nor the controller \eqref{eqn-3}, and it reveals intrinsic properties of the protocol \eqref{eqn-14} itself. Hence, the UGASp protocol provides a novel analysis framework to deal with random packet dropouts in terms of stochastic protocols. On the other hand, Definition \ref{def-2} extends UGAS protocols for the deterministic case \citep{Nesic2004input} to the stochastic case, and includes almost surely Lyapunov UGAS protocols \citep{Tabbara2008input} as a special case.
\hfill $\square$
\end{remark}

For the protocol \eqref{eqn-14}, the following assumption is made, which is based on deterministic protocols \citep{Nesic2004input} and the construction of Lyapunov functions in \citep[Section V]{Heemels2010networked}.

\begin{assumption}
\label{asp-2}
There exist a locally Lipschitz function $W: \mathbb{N}_{+}\times\mathbb{R}^{n_{e}}\rightarrow\mathbb{R}_{+}$, $\alpha_{1W}, \alpha_{2W}\in\mathcal{K}_{\infty}$ and $\lambda\geq1, \rho\in[0, 1)$ such that for all $(i, e)\in\mathbb{N}_{+}\times\mathbb{R}^{n_{e}}$,
\begin{subequations}
\label{eqn-18}
\begin{align}
\label{eqn-18a}
\alpha_{1W}(|e|)\leq W(i, e)&\leq\alpha_{2W}(|e|), \\
\label{eqn-18b}
W(i+1, h(i, e))&\leq\rho W(i, e),  \\
\label{eqn-18c}
W(i+1, e+\mathscr{H}(\mathfrak{h}(i, e), 0))&\leq\lambda W(i, e). 
\end{align}
\end{subequations}
\end{assumption}

In Assumption \ref{asp-2}, \eqref{eqn-18a}-\eqref{eqn-18b} are the same as the conditions for Lyapunov UGAS protocols of deterministic protocols \citep{Nesic2004input, Carnevale2007lyapunov}, and \eqref{eqn-18c} is for the case of the transmission failure. If $W$ is the same as the one for Lyapunov UGAS protocols, then $\lambda$ in \eqref{eqn-18c} equals to 1 simply such that \eqref{eqn-18c} always holds. In this way, the existence of the function $W$ can be verified via Lyapunov UGAS protocols for the deterministic case.

Under Assumption \ref{asp-2}, the following proposition is derived, which shows the UGASp property of the protocol \eqref{eqn-14} if the probabilities of transmission failure and success satisfy certain conditions. This proposition extends Proposition 5.3 in \citep{Tabbara2008input} due to the novel condition \eqref{eqn-18c} for the packet dropout case.

\begin{proposition}
\label{prop-2}
Consider the protocol \eqref{eqn-14}, and let Assumption \ref{asp-2} hold. The protocol \eqref{eqn-14} is UGASp, if
\begin{align}
\label{eqn-19}
\bar{\rho}:=\lambda P_{\mathrm{f}}+\rho P_{\mathrm{s}}<1,
\end{align}
where $P_{\mathrm{f}}$ and $P_{\mathrm{s}}$ are respectively the probabilities of transmission failure and success.
\end{proposition}

\begin{proof}
Based on Definition \ref{def-2}, \eqref{eqn-17a} holds from \eqref{eqn-18a}, and we only need to verify \eqref{eqn-17b}. Since there are only two cases at each transmission time, we have from \eqref{eqn-18b}-\eqref{eqn-18c} that
\begin{align*}
\int_{\upsilon\in\mathcal{V}}W(i+1, \mathbf{h}(i, e(i), \upsilon))\mu(d\upsilon)&\leq(\lambda P_{\mathrm{f}}+\rho P_{\mathrm{s}})W(i, e) \\
&=\bar{\rho}W(i, e).
\end{align*}
From \eqref{eqn-19}, we have that \eqref{eqn-17b} holds and conclude that the protocol \eqref{eqn-14} is UGASp.
\hfill $\blacksquare$
\end{proof}

If $P_{\mathrm{f}}\equiv0$, then Proposition \ref{prop-3} is reduced to the one for the deterministic case. If $\lambda=1$ and $\mathscr{H}(\mathfrak{h}(i, e), 0)=0$, then Proposition \ref{prop-3} is similar to Proposition 5.3 in \citep{Tabbara2008input}. The following proposition establishes the relation between the classic UGAS protocol and the proposed UGASp protocol.

\begin{proposition}
\label{prop-3}
Consider the protocol \eqref{eqn-14}. Let Assumption \ref{asp-2} hold and the mapping $\mathfrak{h}$ depend on a deterministic protocol via \eqref{eqn-8}. If the deterministic protocol is UGAS, then it is a UGASp protocol with $P_{\mathrm{f}}\in[0, \lambda^{-1}(\rho_{1}-\rho))$, where $\rho_{1}\in(\rho, 1)$ and $\lambda, \rho$ are given in Assumption \ref{asp-2}.
\end{proposition}

\begin{proof}
For any deterministic protocol, we have the transmission success with the probability 1, which means $P_{\mathrm{s}}=1$ and $P_{\mathrm{f}}=0$. In addition, if the deterministic protocol is UGAS, then $\lambda=1$. Let $P_{\mathrm{f}}\in[0, (\rho_{1}-\rho)\lambda^{-1})$, and we have $P_{\mathrm{s}}\in(1-(\rho_{1}-\rho)\lambda^{-1}, 1]$ and
\begin{align*}
\int_{\upsilon\in\mathcal{V}}W(i+1, \mathbf{h}(i, e(i), \upsilon))\mu(d\upsilon)&=(\lambda P_{\mathrm{f}}+\rho P_{\mathrm{s}})W(i, e) \\
&\leq(\rho_{1}-\rho+\rho P_{\mathrm{s}})W(i, e) \\
&\leq\rho_{1}W(i, e).
\end{align*}
Therefore, any UGAS protocol is a UGASp protocol with $P_{\mathrm{f}}\in[0, (\rho_{1}-\rho)\lambda^{-1})$, which completes the proof.
\hfill $\blacksquare$
\end{proof}

Proposition \ref{prop-3} implies that any UGAS protocol is a UGASp protocol if the probability of the transmission failure is sufficiently small. In addition, since $\mathfrak{h}$ is required to be related to the deterministic protocol only, Proposition \ref{prop-3} is applicable to many deterministic protocols including those in \citep{Nesic2004input, Maass2019lp}.

\section{Stability Analysis}
\label{sec-stability}

After the analysis of the stochastic protocol, the main result of this paper is presented in this section. To this end, some assumptions on \eqref{eqn-10} are stated, which are adopted from \citep{Carnevale2007lyapunov}.

\begin{assumption}
\label{asp-3}
For the function $W$ in Assumption \ref{asp-2}, there exist a continuous function $H: \mathbb{R}^{n_{x}}\rightarrow\mathbb{R}_{+}$, and $L\in\mathbb{R}_{+}$ such that for all $(x, \kappa)\in\mathbb{R}^{n_{x}}\times\mathbb{N}_{+}$ and almost all $e\in\mathbb{R}^{n_{e}}$,
\begin{align}
\label{eqn-20}
\left\langle\frac{\partial W(\kappa, e)}{\partial e}, g(x, e)\right\rangle&\leq LW(\kappa, e)+H(x).
\end{align}
\end{assumption}

Assumption \ref{asp-3} implies the exponential growth of the function $W$ on the flow, which is natural since the $e$-subsystem in \eqref{eqn-10} is generally unstable in transmission intervals. Assumption \ref{asp-3} is feasible if the function $g$ in \eqref{eqn-12} satisfies a linear growth condition and $W$ is globally Lipschitz in $e$ and uniformly in $i$; see \citep[Remark 11]{Nesic2004input}. For the $x$-subsystem in \eqref{eqn-10}, the following assumption is such that, under the controller \eqref{eqn-3}, the $x$-subsystem is stabilizable with respect to the network-induced error.

\begin{assumption}
\label{asp-4}
There exist a locally Lipschitz function $V: \mathbb{R}^{n_{x}}\rightarrow\mathbb{R}_{+}$, $\alpha_{1V}, \alpha_{2V}\in\mathcal{K}_{\infty}$, continuous functions $\rho: \mathbb{R}^{n_{x}}\rightarrow\mathbb{R}_{+}, \theta: \mathbb{R}_{+}\rightarrow\mathbb{R}_{+}$ and constant $\gamma>0$, such that for all $x\in\mathbb{R}^{n_{x}}$,
\begin{align}
\label{eqn-21}
&\alpha_{1V}(|x|)\leq V(x)\leq\alpha_{2V}(|x|),
\end{align}
and for all $(\kappa, e)\in\mathbb{N}\times\mathbb{R}^{n_{e}}$ and almost all $x\in\mathbb{R}^{n_{x}}$,
\begin{align}
\label{eqn-22}
\langle\nabla V(x), f(x, e)\rangle&\leq-\rho(x)-H^{2}(x)-\theta(W(\kappa, e))+\gamma^{2}W^{2}(\kappa, e),
\end{align}
where $H: \mathbb{R}^{n_{x}}\rightarrow\mathbb{R}_{+}$ is defined in Assumption \ref{asp-3}.
\end{assumption}

Assumption \ref{asp-4} is on the $x$-subsystem, whose properties are described via the function $V$. Under the controller \eqref{eqn-3}, \eqref{eqn-22} implies that the $x$-subsystem satisfies the ISS-like property from $W$ to $x$ and the $\mathcal{L}_2$-stability property from $W$ to $H$. Assumption \ref{asp-4} is reasonable due to the implementation of the emulation-based approach, where the controller is assumed to be known \emph{a priori} to ensure the system stability in the network-free case. In the networked case, $W$ is treated as a disturbance for the $x$-subsystem. Since $W$ depends on the network-induced error and can be treated as a generic perturbation on $(y, u)$, Assumption \ref{asp-4} does not need any knowledge on the network. In this respect, Assumption \ref{asp-4} does not show any stability property of the system \eqref{eqn-10}. The final assumption is on transmission intervals. Different from the deterministic case \citep{Carnevale2007lyapunov}, the transmission failure is involved and the explicit form of the MATI is as follows.

\begin{assumption}
\label{asp-5}
The condition $\rho^{2}P_{\mathrm{s}}+\lambda^{2}P_{\mathrm{f}}<1$ hold, and the MATI is upper bounded as $\tau_{\mati}\leq\tau^{\ast}$, where
\begin{align}
\label{eqn-23}
\tau^{\ast}:=\left\{\begin{aligned}
&\frac{1}{Lr}\arctan\left(\frac{(1-\rho\rho^{\prime})Lr}{\gamma(\rho+\rho^{\prime})+L(1+\rho\rho^{\prime})}\right), & \gamma>L, \\
&\frac{1}{L} \frac{1-\rho\rho^{\prime}}{(1+\rho)(1+\rho^{\prime})}, & \gamma=L, \\
&\frac{1}{Lr} \arctanh\left(\frac{(1-\rho\rho^{\prime})Lr}{\gamma(\rho+\rho^{\prime})+L(1+\rho\rho^{\prime})}\right), & \gamma<L,
\end{aligned}\right.
\end{align}
where $r:=\sqrt{(\gamma/L)^{2}-1}$, $\rho^{\prime}:=\rho^{-1}(\rho^{2}P_{\mathrm{s}}+\lambda^{2}P_{\mathrm{f}})$, and $\rho\in[0, 1), \lambda\geq1, L\geq0, \gamma>0$ are from Assumptions \ref{asp-2}-\ref{asp-4}.
\end{assumption}

The MATI in \eqref{eqn-23} depends on both the constants in Assumptions \ref{asp-2}-\ref{asp-5} and the probabilities of transmission success and failure. The condition $\rho^{2}P_{\mathrm{s}}+\lambda^{2}P_{\mathrm{f}}<1$ is utilised to guarantee the existence of the MATI. This condition can be satisfied if \eqref{eqn-19} holds and $P_{\mathrm{f}}<\frac{1-\rho}{\lambda(\lambda-\rho)}$, which can be treated as an additional constraint on the probability of transmission failure. In particular, this additional constraint can be removed if $\lambda=1$; see Section \ref{sec-casestudy}. On the other hand, \eqref{eqn-23} includes many existing deterministic cases. For instance, if $P_{\mathrm{s}}=1$, then $P_{\mathrm{f}}=0$ and \eqref{eqn-23} is equivalent to the one in \citep{Carnevale2007lyapunov} for the deterministic case. Since $\rho^{2}P_{\mathrm{s}}+\lambda^{2}P_{\mathrm{f}}>\rho^{2}$ holds for all $P_{\mathrm{f}}\in(0, 1)$, the derived MATI in \eqref{eqn-23} is smaller than the one in \citep{Carnevale2007lyapunov}, which implies that the MATI condition for the stochastic case is strictly tighter than the one for the deterministic case.

With all above assumptions, we are now in a position to state the main result of this section. The following theorem and proposition assert that the UGASp protocol leads to the system stability under sufficiently small MATI.

\begin{theorem}
\label{thm-1}
Consider the system \eqref{eqn-10} and let Assumptions \ref{asp-2}-\ref{asp-5} hold. Let $\mathcal{A}:=\{\xi: x=0, e=0\}$. There exists a locally Lipschitz function $U: C\cup D\cup\mathcal{G}(D\times\mathcal{V})\rightarrow\mathbb{R}_{+}$ for the set $\mathcal{A}$ satisfying the following conditions:
\begin{enumerate}[(i)]
  \item  there exist $\alpha_{1}, \alpha_{2}\in\mathcal{K}_{\infty}$ such that for all $\xi\in C\cup D\cup\mathcal{G}(D\times\mathcal{V})$, $\alpha_{1}(|\xi|_{\mathcal{A}})\leq U(\xi)\leq\alpha_{2}(|\xi|_{\mathcal{A}})$;
\vspace{-5pt}
  \item  for all $\xi\in C$ with $U(\xi)\neq0$, $U^{\circ}(\xi; \mathcal{F}(\xi))\leq-\varphi(\xi)$ with $\varphi\in\mathcal{PD}(\mathcal{A})$;
\vspace{-5pt}
  \item  for all $\xi\in D$, $\int_{\upsilon\in\mathcal{V}}U(\mathcal{G}(\xi, \upsilon))\mu(d\upsilon)\leq U(\xi)$.
\end{enumerate}
\end{theorem}

\begin{proof}
See \ref{sec-appendixA}.
\hfill $\blacksquare$
\end{proof}

\begin{proposition}
\label{prop-4}
Consider the system \eqref{eqn-10} and let Assumptions \ref{asp-2}-\ref{asp-5} hold. The set $\mathcal{A}:=\{\xi: x=0, e=0\}$ is UGASp for \eqref{eqn-10}.
\end{proposition}

\begin{proof}
In Theorem \ref{thm-1}, we have shown the existence of the Lyapunov function $U: C\cup D\cup\mathcal{G}(D\times\mathcal{V})\rightarrow\mathbb{R}_{+}$. The set $\mathcal{A}$ is closed by definition but unbounded due to the variable $\kappa\in\mathbb{N}$. Based on \citep[Remark 4]{Nesic2013finite} and item (i) of Theorem \ref{thm-1}, we can see that there exist no finite escape times for the continuous dynamics in \eqref{eqn-10}.
From \citep[Section 4.2]{Teel2013lyapunov} and Theorem \ref{thm-1}, we have that the set $\mathcal{A}$ is UGSp for the system \eqref{eqn-10}.

Next, we show that, given any non-zero level set of $U$, there does not exist an almost surely complete solution remaining in this level set. Let the non-zero level set of $U$ be $\mathcal{Q}_{\delta}:=\{\xi\in\mathbb{R}^{n_{\xi}}: U(\xi)\leq\delta, \delta>0\}$. Assume there exists an almost surely complete solution $\xi_{1}$ remaining in $\mathcal{Q}_{\delta}$ almost surely. If $U(\xi_{1}(t, j))$ converges to zero as $t+j\rightarrow\infty$, then the almost surely complete solution $\xi_{1}$ does not remain in $\mathcal{Q}_{\delta}$ almost surely. Assume there exists $k>t_{0}$ such that for all $(t, j)\in\dom\xi_{1}$ with $t+j\geq k$, $U(\xi_{1}(t, j))\leq\delta$ and $U(\xi_{1}(t, j))\neq0$. In the flow set, we have from item (ii) of Theorem \ref{thm-1} that $U$ is strictly decreasing. That is, for $(t_{j}, j), (t_{j+1}, j)\in\dom\xi_{1}$, $U(\xi_{1}(t_{j+1}, j))<U(\xi_{1}(t_{j}, j))$. In the jump set, $U$ does not increase from item (iii) of Theorem \ref{thm-1}, and $U(\xi_{1}(t_{j+1}, j+1))\leq U(\xi_{1}(t_{j+1}, j))$. Since the jumps occur at least $\tau_{\miati}$ units of time, along the hybrid time line the almost surely complete solution $\xi_{1}$ will enter into a smaller level from $\mathcal{Q}_{\delta}$. That is, there exist $\bar{k}>k$ and $\bar{\delta}<\delta$ such that $\xi_{1}(t, j)\in\mathcal{Q}_{\bar{\delta}}$ holds for all $(t, j)\in\dom\xi_{1}$ with $t+j\geq\bar{k}$. Hence, the almost surely complete solution $\xi_{1}$ will be convergent instead of stay in the level set $\mathcal{Q}_{\delta}$, which contradicts with the aforementioned assumption. As a result, based on the above analysis and \citep[Theorem 8]{Subbaraman2016recurrence}, the set $\mathcal{A}$ is UGASp for the system \eqref{eqn-10}.
\hfill $\blacksquare$
\end{proof}

Theorem \ref{thm-1} and Proposition \ref{prop-4} show that under the stochastic channel, the stability property of the system \eqref{eqn-2}-\eqref{eqn-3} is guaranteed if the MATI is upper bounded in Assumption \ref{asp-5}. In this way, Theorem \ref{thm-1} and Proposition \ref{prop-4} extend the existing result in \citep{Carnevale2007lyapunov} from the deterministic case to the stochastic case.

\section{Case Study: Linear Systems with Specific Networks}
\label{sec-casestudy}

Consider the plant and controller with the following linear state-space form
\begin{align}
\label{eqn-24}
\begin{aligned}
\dot{x}_{\p}&=A_{\p} x_{\p}+B_{\p}\hat{u}, &\quad y&=C_{\p}x_{\p}, \\
\dot{x}_{\ct}&=A_{\ct}x_{\ct}+B_{\ct}\hat{y}, &\quad u&=C_{\ct}x_{\ct}.
\end{aligned}
\end{align}
Let $x=(x_{\p}, x_{\ct})$ and $e=(e_{y}, e_{u})$. The dynamics in the continuous intervals is (see also \citep{Walsh2002stability, Nesic2004input})
\begin{align}
\label{eqn-25}
\dot{x}=A_{1}x+B_{1}e,  \quad \dot{e}=A_{2}x+B_{2}e
\end{align}
with
\begin{align*}
A_{1}&=\begin{bmatrix}A_{\p} & B_{\p}C_{\ct} \\ B_{\ct}C_{\p} & A_{\ct}  \end{bmatrix}, \quad
B_{1}=\begin{bmatrix}0 & B_{\p} \\ B_{\ct}  & 0 \end{bmatrix}, \\
A_{2}&=\begin{bmatrix}-C_{\p}A_{\p} & -C_{\p}B_{\p}C_{\ct} \\ -C_{\ct}B_{\ct}C_{\p} & -C_{\ct}A_{\ct}  \end{bmatrix}, \quad
B_{2}=\begin{bmatrix}0 & -C_{\p}B_{\p} \\ -C_{\ct}B_{\ct}  & 0 \end{bmatrix}.
\end{align*}
In the network-free case, the system \eqref{eqn-25} is assumed to be stable. In the following, we consider two types of networks.

\emph{1) The Ethernet-like Networks \citep{Tabbara2008input}:}
In Ethernet-like networks, carrier sense multiple access with collision detection (CSMA/CD) is a widely-used medium access protocol allowing different links to access to the channel if the channel is idle. Therefore, the packet dropout occurs when multiple links attempt to transmit their packets. To model the CSMA/CD protocol, the following stochastic protocol is proposed in \citep{Tabbara2008input}:
\begin{equation}
\label{eqn-26}
h(i, e):=\mathcal{H}_{i}(e)e,
\end{equation}
where $\mathcal{H}_{i}(\cdot)$ is an i.i.d. random mapping and takes values from a finite matrix set $\mathcal{M}:=\{M_{0}, M_{1}, \ldots, M_{\ell}\}$ with $M_{0}=I$. Here, $\ell$ is the number of the links. Each component in $\mathcal{M}$ corresponds to a possible case of the transmission success via certain link \citep[Section IV]{Tabbara2008input}. $M_{l}$ is such that $M_{l}e=(e_{1}, \ldots, e_{l-1}, 0, e_{l+1}, \ldots, e_{\ell})$, which means that the $l$-th link acquires the channel and transmits its packet successfully. For $i\in\mathbb{N}_{+}$ and $j\in\{1, \ldots, \ell\}$, $\pr\{\mathcal{H}_{i}(e)=M_{j}\}=P_{\mathrm{s}}/\ell$, which implies that each link is equally likely to be transmitted successfully. Therefore, the matrix choice from the set $M$ consists of two parts: the choice of the node and the transmission success. Let the choice of the node be deterministic, and thus both Round-Robin (RR) protocol and Try-Once-Discard (TOD) protocol can be applied \citep{Nesic2004input}.

To develop the stochastic hybrid model, we introduce a random variable $\upsilon\in\{0, 1\}^{\ell}$ to show the transmission status of each link, and rewrite \eqref{eqn-26} into the following stochastic protocol:
\begin{align}
\label{eqn-27}
\mathbf{h}(i, e, \upsilon)=(I-\mathscr{S}(\upsilon))e, 
\end{align}
which implies that $\mathscr{H}(\mathfrak{h}(i, e(i)), \mathscr{S}(\upsilon(i+1)))=0$. In \eqref{eqn-27}, the random variable $\upsilon$ plays the same role as the matrix set $\mathcal{M}$, and $\pr\{\upsilon_{l}=1\}=P_{\mathrm{s}}/\ell$. From the above analysis for the Ethernet-like network, the overall system dynamics is
\begin{align}
\label{eqn-28}
\begin{aligned}
\dot{x}&=A_{1}x+B_{1}e, \quad t\in[t_{i}, t_{i+1}], \\
\dot{e}&=A_{2}x+B_{2}e, \quad t\in[t_{i}, t_{i+1}], \\
e(t^{+}_{i})&=(I-\mathscr{S}(\upsilon(t^{+}_{i})))e(t_{i}),   \\
\upsilon&\sim\mu(\cdot).
\end{aligned}
\end{align}

Since the system \eqref{eqn-24} is stable in the network-free case, the $x$-subsystem is  $\mathcal{L}_{2}$-gain stable from $e$ to $A_{1}x$ with gain $\gamma\geq0$, which means the satisfaction of Assumption \ref{asp-4}. Next, we state a proposition to show the satisfaction of Assumptions \ref{asp-2}-\ref{asp-3} and further the stability property of the system \eqref{eqn-28}.

\begin{proposition}
\label{prop-5}
Consider the system \eqref{eqn-28}.
\begin{enumerate}[(1)]
  \item If the node choosing strategy is of the RR mechanism, then $\alpha_{1W}=1, \alpha_{2W}=\sqrt{\ell}, \rho=\sqrt{(\ell-1)/\ell}, \lambda=\sqrt{\ell}, L=\sqrt{\ell}|B_{1}|$ such that Assumptions \ref{asp-2}-\ref{asp-3} hold. In addition, if $(1-P_{\mathrm{s}})\ell+P_{\mathrm{s}}(\ell-1)/\ell<1$, where $P_{\mathrm{s}}\in(0, 1)$ is the successful transmission probability, then the system \eqref{eqn-28} is UGASp with the MATI defined in \eqref{eqn-23}.\vspace{-5pt}

  \item If the node choosing strategy is of the TOD mechanism, then $\alpha_{1W}=\alpha_{2W}=1, \rho=\sqrt{(\ell-1)/\ell}, \lambda=1, L=|B_{1}|$ such that Assumptions \ref{asp-2}-\ref{asp-3} hold. In addition, the system \eqref{eqn-28} is UGASp with the MATI defined in \eqref{eqn-23}.
\end{enumerate}
\end{proposition}

\newcounter{TempEqCnt}
\setcounter{TempEqCnt}{\value{equation}} 
\setcounter{equation}{30} 
\begin{figure*}[!b]
\normalsize
\hrulefill
\begin{align}
\label{eqn-31}
\begin{aligned}
A_{1}&=\begin{bmatrix}1.3800 & -0.2077 & 6.7150& -5.6760 & 0& 0 \\ -0.5814 &-15.6480 & 0& 0.6750 &-11.3580&0 \\
14.6630 &2.0010 &-22.3840 & 21.6230 &-2.2720 &-25.1680\\ 0.048 &4.273 & 1.3430 &-2.1040 &-2.2720 &0  \\
0 &1.0000&0&0&0&0 \\ 1.000 &0 &1.0000&-1.0000& 0&0\end{bmatrix}, \quad B_{1}=\begin{bmatrix}0 & 0 \\
0 &-11.3580 \\ -15.7300 &-2.2720 \\ 0 &-2.2720  \\ 0 &1.0000 \\ 1.000 &0 \end{bmatrix}, \\
A_{2}&=\begin{bmatrix} 13.3310&0.2077&17.0120&-18.0510&0& 25.1680 \\
0.5814 &15.6480&0 & -0.6750 &11.3580 & 0 \end{bmatrix}, \quad B_{2}=\begin{bmatrix}15.7300 & 0 \\ 0 & 11.3580 \end{bmatrix}.
\end{aligned}
\end{align}
\end{figure*}
\setcounter{equation}{\value{TempEqCnt}} 

In item (2) of Proposition \ref{prop-5}, since $\rho^{2}P_{\mathrm{s}}+\lambda^{2}P_{\mathrm{f}}=(1-P_{\mathrm{s}})+P_{\mathrm{s}}(\ell-1)/\ell<1$ holds directly, there exists no constraint on the probability of the transmission success in the TOD case.

\emph{2) WirelessHART Networks \citep{Maass2019stabilization}: }
To model WirelessHART (WH) networks, the stochastic protocol of the form \eqref{eqn-26} is proposed in \citep{Maass2019stabilization}. The difference between the protocols in \citep{Tabbara2008input} and \citep{Maass2019stabilization} lies in the matrix set $\mathcal{M}$. In \citep{Maass2019stabilization}, the WH network has $\ell_{y}$ field devices in the plant-to-controller side and $\ell_{u}$ field devices in the controller-to-plant side. Each field device has both the reception and transmission behaviour. Considering the transmission between the field devices and the plant/controller, $\mathcal{M}\setminus\{M_{0}\}$ has $\ell_{y}+\ell_{u}+2$ elements. Each component in $\mathcal{M}$ corresponds to a possible case of the reception/transmission behaviour of certain field device \citep[Section IV]{Maass2019lp}. Due to both reception and transmission behaviour of field devices, the matrix $M_{l}$ is the form $\diag(M^{y}_{l}, M^{u}_{l})$ with
\begin{align*}
M^{\star}_{l}&=\begin{bmatrix} \Psi_{\star}(l) & 0 \\ I_{\star}-\Psi_{\star}(l) & \Gamma_{\star}(l) \end{bmatrix}, \\
\Gamma^{\star}(l)&=\begin{bmatrix}
\gamma_{1}^{\star}(l) I_{\star} & & & \\
(1-\gamma_{1}^{\star}(l)) I_{\star} & \gamma_{2}^{\star}(l) I_{\star} &  & \\
&  \ddots & \ddots & \\
& & (1-\gamma^{\star}_{\ell_{\star}-1}(l)) I_{\star} & \gamma^{\star}_{\ell_{\star}}(l) I_{\star}
\end{bmatrix},
\end{align*}
where $\star\in\{y, u\}$ and $\Psi_{\star}(l)=\diag(\delta_{1}^{\star}(l) I_{\star}, \ldots, \delta_{\ell_{\star}}^{\star}(l)I_{\star})$. More precisely, for all $\mathfrak{a}\in\{1, \ldots, \ell_{y}\}$ and $\mathfrak{b}\in\{1, \ldots, \ell_{u}\}$, $\delta_{\mathfrak{a}}^{y}(l)=0$ if $l=\mathfrak{a}$; otherwise $\delta_{\mathfrak{a}}^{y}(l)=1$. $\gamma_{\mathfrak{a}}^{y}(l)=0$ if $l=\mathfrak{a}+1$; otherwise, $\gamma_{\mathfrak{a}}^{y}(l)=1$. $\delta_{\mathfrak{b}}^{u}(l)=0$ if $l=\mathfrak{b}+\ell_{y}+1$; otherwise, $\delta_{\mathfrak{b}}^{u}(l)=1$. $\gamma_{\mathfrak{b}}^{u}(l)=0$ if $l=\mathfrak{b}+\ell_{y}+2$; otherwise $\gamma_{\mathfrak{b}}^{u}(l)=1$. From the matrix $M_{l}$, we can see that: (1) for each field device, the updates of the network-induced error during reception and transmission affect each other; (2) the update of the network-induced error in any field device affects the updates in other field devices.

From the above observation, the updates in field devices are coupled, which shows the necessity of the mapping $\mathscr{H}$ in \eqref{eqn-16}. For WH networks, the protocol in \citep{Maass2019stabilization} can be rewritten into the stochastic protocol of the form \eqref{eqn-14} with $\mathbf{h}(i, e, \upsilon):=\diag(\mathbf{h}_{y}(i, e, \upsilon), \mathbf{h}_{u}(i, e, \upsilon))$ and
\begin{align}
\label{eqn-29}
&\mathbf{h}_{\star}(i, e, \upsilon)=(I-\mathscr{S}(\upsilon_{\star}))e+\begin{bmatrix} 0 & 0 \\
\mathscr{S}(\upsilon_{\star}) & \Upsilon_{\star}(\upsilon) \end{bmatrix}e,  \end{align}\begin{align}
&\Upsilon_{y}(\upsilon):=\begin{bmatrix}
(\upsilon^{y}_{1}-\upsilon^{y}_{2})I_{y} & & & \\
\upsilon^{y}_{2}I_{y} & (\upsilon^{y}_{2}-\upsilon^{y}_{3})I_{y} &  & \\
&  \ddots & \ddots & \\
& & \upsilon_{\ell_{y}}I_{y} & (\upsilon_{\ell_{y}}-\upsilon_{\ell_{y}+1})I_{y}
\end{bmatrix}, \nonumber \\
&\Upsilon_{u}(\upsilon):=\begin{bmatrix}\begin{smallmatrix}
(\upsilon^{u}_{1}-\upsilon^{u}_{3})I_{u} & & & & \\
\upsilon^{u}_{3}I_{u} & (\upsilon^{u}_{2}-\upsilon^{u}_{4})I_{u} & & & \\
&   \ddots & \ddots &  &\\
& \upsilon^{u}_{\ell_{u}-1}I_{u} & (\upsilon^{u}_{\ell_{u}-2}-\upsilon^{u}_{\ell_{u}})I_{u} && \\
& & \upsilon^{u}_{\ell_{u}}I_{u} & (\upsilon^{u}_{\ell_{u}-1}-\upsilon_{\ell_{y}+1})I_{u}  &\\
& & & \upsilon_{\ell_{y}+1}I_{u} & (\upsilon^{u}_{\ell_{u}}-\upsilon_{\ell_{u}+1})I_{u}
\end{smallmatrix}\end{bmatrix}, \nonumber
\end{align}
where $\upsilon:=(\upsilon_{y}, \upsilon_{u}, \upsilon_{\ell_{y}+1}, \upsilon_{\ell_{u}+1})\in\mathbb{R}^{\ell_{y}+\ell_{u}+2}$ with $\upsilon_{y}:=(\upsilon^{y}_{1}, \ldots, \upsilon^{y}_{\ell_{y}})\in\mathbb{R}^{\ell_{y}}$ and $\upsilon_{u}:=(\upsilon^{u}_{1}, \ldots, \upsilon^{u}_{\ell_{u}})\in\mathbb{R}^{\ell_{u}}$. In particular, $\upsilon_{\ell_{y}+1}\in\mathbb{R}$ and $\upsilon_{\ell_{u}+1}\in\mathbb{R}$ are introduced to denote the transmission from the plant/controller to the field device and from the field device to the plant/controller, respectively.

From the above analysis, the overall system dynamics is
\begin{align}
\label{eqn-30}
\begin{aligned}
\dot{x}&=A_{1}x+B_{1}e, \quad t\in[t_{i}, t_{i+1}], \\
\dot{e}&=A_{2}x+B_{2}e, \quad t\in[t_{i}, t_{i+1}], \\
e(t^{+}_{i})&=\diag(\mathbf{h}_{y}(e(t_{i}), \upsilon(t^{+}_{i})), \mathbf{h}_{u}(e(t_{i}), \upsilon(t^{+}_{i}))),   \\
\upsilon&\sim\mu(\cdot),
\end{aligned}
\end{align}
where $\mathbf{h}_{y}(e(t_{i}), \upsilon(t^{+}_{i}))$ and $\mathbf{h}_{u}(e(t_{i}), \upsilon(t^{+}_{i}))$ are given in \eqref{eqn-29}. Similar to the case of Ethernet-like Networks, we can guarantee the satisfaction of the proposed assumptions and the stability property of the system \eqref{eqn-30}, which is omitted here.

\section{Numerical Results}
\label{sec-example}

In this section, we first illustrate how to verify the quantitative tradeoff between the MATI and the probabilities of the transmission success for the case study of the batch reactor in Section \ref{subsec-example1}, and then consider a nonlinear example from a single-link robot arm in Section \ref{subsec-example2}. 

\subsection{Batch Reactor}
\label{subsec-example1}

This case study has been considered over the years as a benchmark example in NCS; see, e.g., \citep{Donkers2011stability, Carnevale2007lyapunov, Heemels2010networked, Nesic2004input, Nesic2009unified}. The functions in \eqref{eqn-9} for the batch reactor are given by the linear functions $f(x, e)=A_{1}x+B_{1}e$ and $g(x, e)=A_{2}x+B_{2}e$. The batch reactor has $n_{u}=2$ inputs, $n_{y}=2$ outputs, $n_{\p}=4$ plant states and $n_{\ct}=2$ controller states and $l=2$ nodes. The open-loop system is unstable, and we assume that only the outputs are sent via the network; see \citep{Nesic2004input} for more details on this example. As provided in \citep{Nesic2004input, Heemels2010networked}, the matrices $A_{i}, B_{i}$ with $i=1, 2$ are presented in \eqref{eqn-31}. In the following, both the TOD and RR cases are investigated.

In the TOD case, for the $x$-subsystem, we compute the $\mathcal{L}_{2}$ gain from the input $e$ to the output $A_{2}x$, and derive $\gamma=15.9222$. From \citep{Nesic2004input, Heemels2010networked}, $W(\kappa, e)=|e|$, and thus the gain from $W$ to $A_{2}x$ is also bounded via $\gamma$. In addition, from Proposition \ref{prop-5}, we have $\rho=\sqrt{0.5}, L=15.7300, \lambda=1$. From Propositions \ref{prop-2} and \ref{prop-4}, the protocol and the closed-loop system are UGASp without any constraint on the probability of transmission success.

\begin{figure}[!t]
\begin{center}
\begin{picture}(65, 88)
\put(-65, -22){\resizebox{60mm}{40mm}{\includegraphics[width=2.5in]{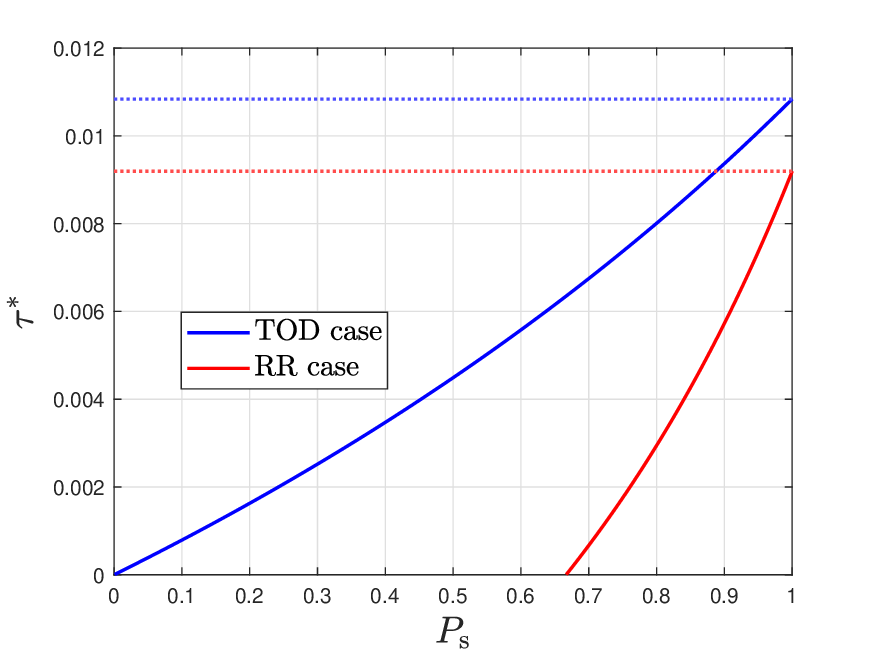}}}
\end{picture}
\end{center}
\caption{Evolution of the MATI bound $\tau^{\ast}$ with different values of $P_{\mathrm{s}}$ in the RR and TOD cases. The dotted lines are for the deterministic case and the solid curves are for the stochastic case.}
\label{fig-1}
\end{figure}

\begin{figure}[!t]
\begin{center}
\begin{picture}(65, 88)
\put(-65, -20){\resizebox{60mm}{40mm}{\includegraphics[width=2.5in]{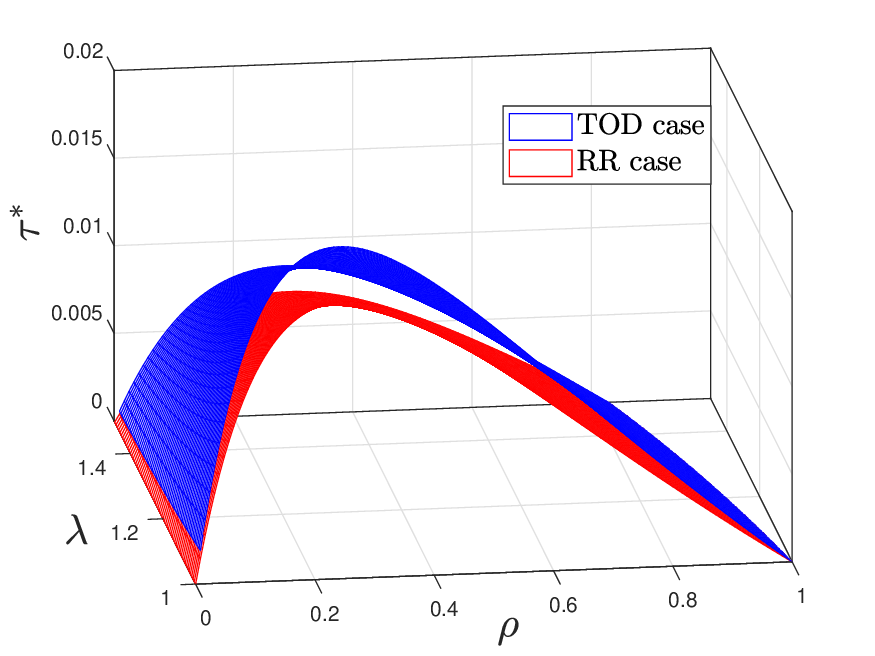}}}
\end{picture}
\end{center}
\caption{Evolution of the MATI bound $\tau^{\ast}$ with different $\lambda$ and $\rho$ in the RR and TOD cases. Here the probability of transmission success is fixed as $0.8$.}
\label{fig-2}
\end{figure}

In the RR case, we follow the similar fashion in the TOD case to compute the $\mathcal{L}_{2}$ gain from the input $e$ to the output $\sqrt{2}A_{2}x$, that is, $\gamma=15.2222\sqrt{2}$. For the RR case, $W(\kappa, e)=\sqrt{a^{2}_{1}(i)e^{2}_{1}+a^{2}_{2}(i)e^{2}_{2}}$, where $a^{2}_{j}(i)=1$ or $2$ for $j=1, 2$ and $i\in\mathbb{N}$. That is, $\max\{a_{1}(i), a_{2}(i)\}\leq\sqrt{2}$, which further implies $\lambda=\sqrt{2}$. From \citep{Heemels2010networked}, we have $\rho=\sqrt{0.5}$ and $L=15.73$. From Proposition \ref{prop-2}, the protocol is UGASp if $P_{\mathrm{s}}>2-\sqrt{2}$. In addition, from Proposition \ref{prop-5}, the closed-loop system is UGASp if $P_{\mathrm{s}}>2/3$. Hence, $P_{\mathrm{s}}>\max\{2-\sqrt{2}, 2/3\}=2/3$.

For both the TOD and RR cases, $\gamma>L$, and from \eqref{eqn-23}, the relation between the MATI bound $\tau^{\ast}$ and the probability $P_{\mathrm{s}}$ is depicted in Figure \ref{fig-1}. Since no constraints are imposed to the probability of transmission success in the TOD case, the relation between $\tau^{\ast}$ and $P_{\mathrm{s}}$ is continuous in $(0, 1)$. The larger the probability $P_{\mathrm{s}}$ is, the larger the MATI bound $\tau^{\ast}$. However, due to $P_{\mathrm{s}}>2/3$ in the RR case, we can see from Figure \ref{fig-1} that $\tau^{\ast}>0$ for $P_{\mathrm{s}}\in(2/3, 1)$. When $P_{\mathrm{s}}=1$, the derived MATI bound $\tau^{\ast}$ is the same as the one for the deterministic case as in \citep{Nesic2004input}.

On the other hand, let the probability of the transmission success be fixed. Here we set $P_{\mathrm{s}}=0.8$, which shows that $\tau^{\ast}>0$ holds for both the RR and TOD cases. With the fixed $P_{\mathrm{s}}$, we investigate the relation between the MATI bound and the parameters $\rho, \lambda$ in \eqref{eqn-18b}-\eqref{eqn-18c}. In this case, the evolution of the MATI bound $\tau^{\ast}$ is depicted in Figure \ref{fig-2}. In particular, if $\rho\in(0, 1)$ is fixed, then $\tau^{\ast}$ decreases with the increase of $\lambda$, which implies the effects of the update at the transmission failure. If $\lambda>1$ is fixed, then $\tau^{\ast}$ increases first and then decreases with the increase of $\rho$, which shows the effects of the update at the transmission success.

\subsection{Single-Link Robot Arm}
\label{subsec-example2}

\setcounter{equation}{31}
Consider a single-link robot arm with the following dynamics (see also \citep{Postoyan2014tracking})
\begin{align}
\label{eqn-32}
\dot{x}_{1}=2.5x_{2}, \quad \dot{x}_{2}=-a\sin(x_{1})+bu, 
\end{align}
where $x_{1}\in\mathbb{R}$ is the angle, $x_{2}\in\mathbb{R}$ is the rotational velocity, $u$ is the input torque, and $a, b>0$ are the fixed parameters. Both the angle and rotational velocity are measured. In the network-free case, the controller $u=b^{-1}(a\sin(x_{1})-x_{1}-0.5x_{2})$ is designed such that $(x_{1}, x_{2})$ converges asymptotically towards the origin. 

In the following, we address the case where the communication between the plant and the controller is via a network. In the networked case, the emulated controller is $u=b^{-1}(a\sin(\hat{x}_{1})-\hat{x}_{1}-\hat{x}_{2})$, where $\hat{x}_{1}, \hat{x}_{2}\in\mathbb{R}$ are implemented in the ZOH fashion. Let the network have 3 nodes for $x_{1}, x_{2}$ and $u$, respectively. Hence, $\ell=3$, and each node has a binary variable $\upsilon_{1}, \upsilon_{2}, \upsilon_{3}\in\{0, 1\}$ to show whether the packet dropout occurs. For these three nodes, let $\pr\{\upsilon_{l}=1\}=p_{l}\in(0, 1)$ be the probability of the transmission success. 

\begin{figure}[!t]
\begin{center}
\begin{picture}(65, 88)
\put(-65, -22){\resizebox{65mm}{40mm}{\includegraphics[width=2.5in]{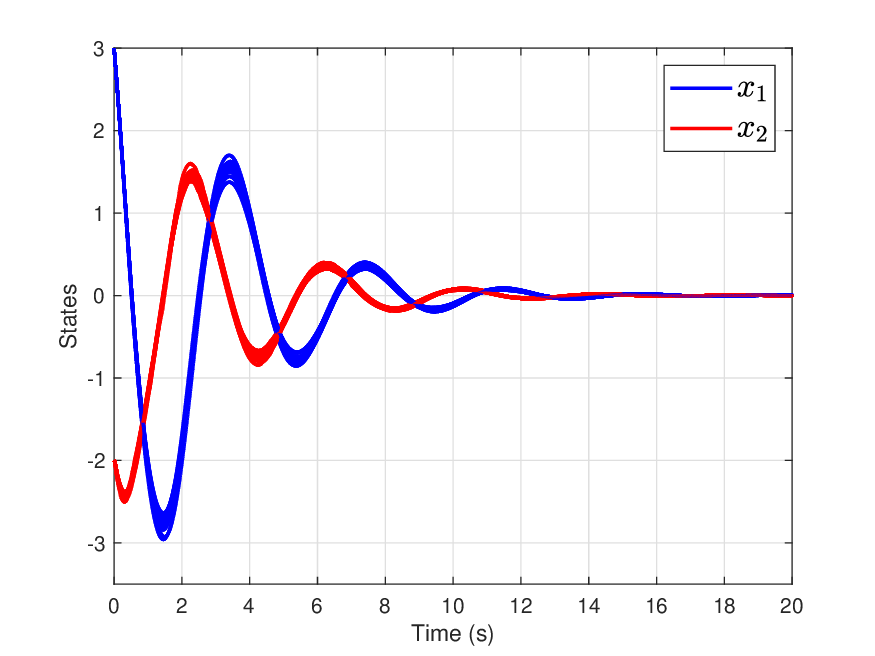}}}
\end{picture}
\end{center}
\caption{Illustration of 15 realizations of the  state trajectories of the system \eqref{eqn-32} in the TOD case.}
\label{fig-3}
\end{figure}

\begin{figure}[!t]
\begin{center}
\begin{picture}(65, 88)
\put(-65, -22){\resizebox{65mm}{40mm}{\includegraphics[width=2.5in]{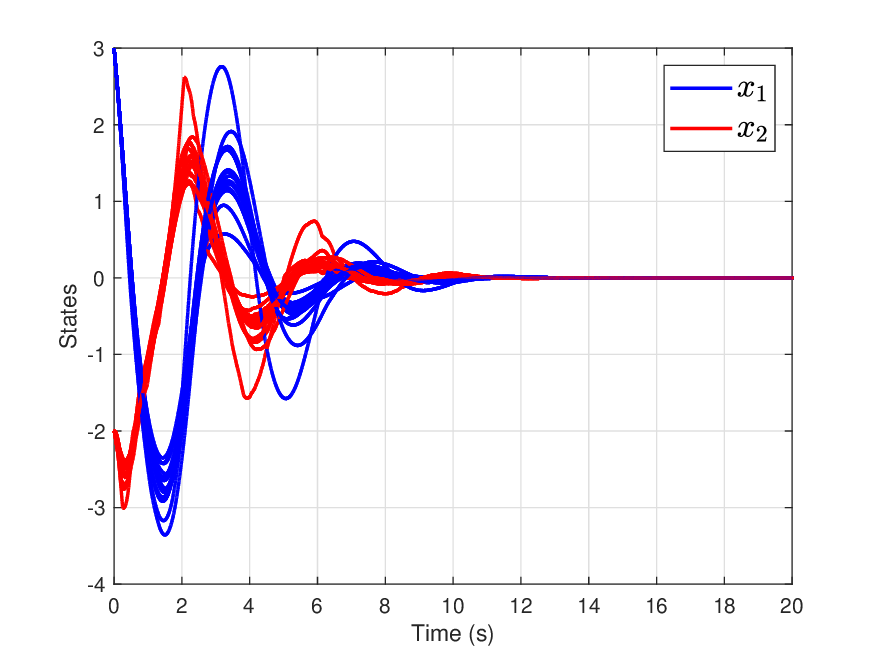}}}
\end{picture}
\end{center}
\caption{Illustration of 15 realizations of the state trajectories of the system \eqref{eqn-32} in the RR case.}
\label{fig-4}
\end{figure}

Next we investigate both the RR and TOD cases to verify Assumptions \ref{asp-2}-\ref{asp-5}. For this purpose, we derive the system model \eqref{eqn-9}. For the continuous-time interval, we have $f(x, e)=(2.5x_{2}, -a(\sin(x_{1})-\sin(x_1+e_{1}))-(x_1+e_{1})-0.5(x_{2}+e_{2})+be_{u})$ and $g(x, e)=(-f(x, e), 0)$. At the transmission time instants, 
\begin{align*}
e(t^{+}_{i})=(1-\bar{\upsilon}(t^{+}_{i}))e(t_{i})+\bar{\upsilon}(t^{+}_{i})h(i, e(t_{i})), 
\end{align*}
where $\bar{\upsilon}(t)=\min\{\upsilon_{1}(t), \upsilon_{2}(t), \upsilon_{3}(t)\}$ and $h(i, e)$ is the same as the one in \eqref{eqn-5}. That is, if one node is chosen and no packet dropout occurs, then the update of the network-induced error is the same as \eqref{eqn-5}. For the RR case, the function $W$ is defined as $W(\kappa, e)=\sqrt{\sum_{i=\kappa}^{\infty}|\phi(i, \kappa, e)|^2}$, where $\phi(i, \kappa, e)$ is the solution to $e^{+}=h(i, e)$ at time $i$ starting at time $\kappa$ with the initial condition $e$. From \citep[Proposition 4]{Nesic2004input}, we have that Assumption \ref{eqn-2} holds with $\alpha_{1W}(s)=s$, $\alpha_{2W}(s)=\sqrt{3}s$, $\rho=\sqrt{2/3}$ and $\lambda=\sqrt{3}$. For the TOD case, we set $W(e):=|e|$ directly, and thus Assumption \ref{eqn-2} holds with $\alpha_{1W}(s)=\alpha_{2W}(s)=s$, $\rho=\sqrt{2/3}$ and $\lambda=1$. In addition, we have that $|g(x, e)|\leq2.5|x_{2}|+|x_{1}+x_{2}|+D|e|$, where $D:=\sqrt{3}\max\{1+a, b\}$. From \citep[Lemma V.4]{Heemels2010networked}, $|\partial W(\kappa, e)/\partial e|\leq M$ for all $\kappa\in\mathbb{N}$ and  almost all $e\in\mathbb{R}^{3}$. Here, $M=\sqrt{3}$ for the RR case and $M=1$ for the TOD case. Hence, $|\langle\partial W(\kappa, e)/\partial e, g(x, e)\rangle|\leq M(DW(\kappa, e)+2.5|x_{2}|+|x_{1}+x_{2}|)$ for almost all $e\in\mathbb{R}^{3}$ and all $(x, \kappa)\in\mathbb{R}^{2}\times\mathbb{N}$. That is, Assumption \ref{asp-3} holds with $L:=MD$ and $H(x):=M(2.5|x_{2}|+|x_{1}+x_{2}|)$. 

In order to show Assumption \ref{asp-4}, we define $V(x):=\mathfrak{a}x_{1}^2+\mathfrak{b}x_{1}x_{2}+\mathfrak{c}x_{2}^2$, where $\mathfrak{a}, \mathfrak{b}, \mathfrak{c}\in\mathbb{R}$ are such that \eqref{eqn-20} holds. Let $a(\sin(x_{1})-\sin (x_{1}+e_{1}))=\bar{a}e_{1}$ with a varying parameter $\bar{a}\in[-a, a]$. Hence, $\langle\nabla V(x), f(x, e)\rangle\leq-\mathfrak{b}x^2_{1}-(2\mathfrak{c}-\mathfrak{b})x^2_{2}+(2\mathfrak{a}
-2\mathfrak{c}-\mathfrak{b})x_{1}x_{2}+(\mathfrak{b}x_{1}+2\mathfrak{c}x_{2})((-\bar{a}-1)e_{1}-e_{2}+be_{u})$. Using the fact that $xy\leq\eta x^2/2+y^2/(2\eta)$ for $x, y\in\mathbb{R}_{+}$ and $\eta>0$, we derive that $\langle\nabla V(x), f(x, e)\rangle\leq-\mathfrak{b}x^2_{1}-(2\mathfrak{c}-\mathfrak{b})x^2_{2}+(2\mathfrak{a}-2\mathfrak{c}-\mathfrak{b})x_{1}x_2+(2\mathfrak{c}x_{2}+\mathfrak{b}x_{1})^2/(2\eta)+\eta D^2|e|^2/2$, where $\eta>0$ is arbitrary. Hence, if $\mathfrak{a}, \mathfrak{b}, \mathfrak{c}\in\mathbb{R}$ are such that \eqref{eqn-20} is satisfied and $-\varepsilon|x|^2-H^2(x)\geq-\mathfrak{b}x_{1}^2-(2\mathfrak{c}-\mathfrak{b})x_{2}^2 +(2\mathfrak{a}-2\mathfrak{c}-\mathfrak{b}) x_{1}x_{2}+(2\mathfrak{c}x_{2}+\mathfrak{b}x_{1})^2/(2\eta)$, where $\varepsilon>0$ is arbitrary, then we conclude that Assumption \ref{asp-4} holds with $\gamma:=\sqrt{\eta D^2/2+\varepsilon}$ and $\theta(s):=\varepsilon s^{2}$. 

Let $a=0.5\cdot9.81, b=2$ and we set $\mathfrak{a}=8, \mathfrak{b}=6, \mathfrak{c}=4, \eta=7$ and $\varepsilon=0.005$. We can compute that $L=17.7150, \gamma=19.1344$ for the RR case and $L=10.2278, \gamma=19.1345$ for the TOD case. For the three nodes, the probabilities of the transmission success are set as $0.2, 0.4, 0.6$, respectively. Hence, at each transmission time, the probability of the transmission success is more than 0.808. From \eqref{eqn-23}, the upper bound of the MATI is $\tau^{\ast}=0.006874$ for the TOD case and $\tau^{\ast}=0.005482$ for the RR case. Let $\tau_{\mati}=0.005$, and Figures \ref{fig-3}-\ref{fig-4} show the convergence of the system state in the TOD and RR cases, respectively. From 15 different realizations, we can see the guarantee of the stochastic stability property of the closed-loop system with stochastic channels.

\section{Conclusion}
\label{sec-conclusion}

In this paper, we proposed a novel stochastic hybrid model for networked control systems with random packet dropouts, and then provided conditions to guarantee uniform global asymptotic stability in probability of the whole system. In particular, the transmission times are deterministic whereas the dropouts of the transmitted information are stochastic, which results in a novel stochastic scheduling protocol. Based on stochastic hybrid system theory, the proposed model was amenable to elegant analysis, which generalizes the existing results in \citep{Carnevale2007lyapunov} directly. Future work will focus on specific networks and more general cases including networked control systems with time delays and the case of discrete-time controllers.

\appendix
\section{Proof of Theorem \ref{thm-1}}
\label{sec-appendixA}

\setcounter{equation}{0}
\renewcommand{\theequation}{A.\arabic{equation}}

Let $\phi: [0, \tau^{\ast}]\rightarrow\mathbb{R}$ be the solution to
\begin{align}
\label{eqn-A1}
\dot{\phi}=-2L\phi-\gamma(\phi^{2}+1), \quad \phi(0)=\rho^{-1},
\end{align}
where $L\geq0$ and $\gamma>0$ are from Assumptions \ref{asp-3}-\ref{asp-4}, respectively. Note that $\tau^{\ast}$ in \eqref{eqn-23} corresponds to the time it takes for the solution to \eqref{eqn-A1} to decrease from the initial condition $\phi(0)=\rho^{-1}$ to $\phi(\tau^{\ast})=\rho^{\prime}$, where $\rho^{\prime}$ is defined below \eqref{eqn-23}. Based on this goal and \citep[Lemma 2]{Carnevale2007lyapunov}, the following claim is derived directly.

\begin{claim}
\label{claim-1}
For all $\tau\in[0, \tau^{\ast}]$, $\phi(\tau)\in[\rho^{\prime}, \rho^{-1}]$.
\end{claim}

Define the following function for the system \eqref{eqn-10}:
\begin{equation}
\label{eqn-A2}
U(\xi):=V(x)+\gamma\phi(\tau)W^{2}(\kappa, e).
\end{equation}
From \eqref{eqn-A1}, Assumptions \ref{asp-2}-\ref{asp-4} and the continuity of $F$ and $G$ in \eqref{eqn-10}, it is easy to check that $U$ in \eqref{eqn-A2} is continuous and locally Lipschitz. In the following, we first show that $U(\xi)$ is suitable Lyapunov function for the system \eqref{eqn-10} (i.e., \textbf{Step 1}), and then derive the decrease of $U(\xi)$ in the hybrid time domain (i.e., \textbf{Step 2}). In this way, we conclude that all three conditions in Theorem \ref{thm-1} are satisfied for the function $U$ in \eqref{eqn-A2}.

\textbf{Step 1: Positive Definiteness and Radial Unboundedness of $U(\xi)$.} For the Lyapunov function $U$ defined in \eqref{eqn-A2}, we obtain from Assumptions \ref{asp-2}-\ref{asp-4} that
\begin{align*}
\alpha_{1V}(|x|)+\gamma\phi(\tau_{\mati})\alpha^{2}_{1W}(|e|)\leq U(\xi)&\leq\alpha_{2V}(|x|)+\gamma\phi(0)\alpha^{2}_{2W}(|e|).
\end{align*}
From \eqref{eqn-22}, we have $\phi(\tau_{\mati})\geq(\lambda^{2}P_{\mathrm{f}}+\rho^{2}P_{\mathrm{s}})\phi(0)$. Furthermore, from \citep[Remark 2.3]{Laila2002lyapunov}, there exist $\alpha_{1}, \alpha_{2}\in\mathcal{K}_{\infty}$ such that
\begin{align}
\label{eqn-A3}
\alpha_{1}(|\xi|_{\mathcal{A}})\leq U(\xi)\leq\alpha_{2}(|\xi|_{\mathcal{A}}), \quad \xi\in C\cup D\cup\mathcal{G}(D\times\mathcal{V}),
\end{align}
where $\alpha_{1}(v):=\min\{\alpha_{1V}(v/2), (\lambda^{2}P_{\mathrm{f}}+\rho^{2}P_{\mathrm{s}})\gamma\phi(0)\alpha^{2}_{1W}(v/2)\}$ and
$\alpha_{2}(v):=\max\{\alpha_{2V}(v), \gamma\phi(0)\alpha^{2}_{2W}(v)\}$. Hence, item (i) is valid.

\textbf{Step 2: Decreasing of $U(\xi)$ on the Flow and at Jumps.} Consider the evolution of $U$ in the continuous-time intervals. For all $\xi\in C$, the differential\footnote{Because of $\kappa$, $U$ is not differential almost everywhere. Since $\dot{\kappa}=0$ in \eqref{eqn-12}, we can consider the differential of $U$ with a slight abuse of terminology.} of $U(\xi)$ satisfies
\begin{align}
\label{eqn-A4}
&U^{\circ}(\xi; \mathcal{F}(\xi))  \nonumber\\
&\leq-\rho(x)-\theta(W(\kappa, e))-H^{2}(x)+\gamma^{2}W^{2}(\kappa, e)  \nonumber\\
&\quad +\gamma\dot{\phi}(\tau)W^{2}(\kappa, e)+2\gamma\phi(\tau)W(\kappa, e)[LW(\kappa, e)+H(x)] \nonumber\\
&=-\rho(x)-\theta(W(\kappa, e))-[H(x)-\gamma\phi(\tau)W(\kappa, e)]^{2}  \nonumber\\
&\leq-\rho(x)-\theta(W(\kappa, e)),
\end{align}
where the first ``$\leq$'' is from Assumptions \ref{asp-3}-\ref{asp-4}. From \eqref{eqn-A4}, $U(\xi)$ decreases strictly in the flow set, which implies the existence of the function $\varphi\in\mathcal{PD}(\mathcal{A})$ such that item (ii) is satisfied. For all $\xi\in D$, from Assumptions \ref{asp-2} and \ref{asp-5}, the jump of $U(\xi)$ satisfies
\begin{align}
\label{eqn-A5}
&\int_{\upsilon\in\mathcal{V}}U(\mathcal{G}(\xi, \upsilon))\mu(d\upsilon)  \nonumber\\
&=\int_{\upsilon\in\mathcal{V}}(V(x)+\gamma\phi(0)W(\kappa+1, \mathbf{h}(\kappa, e, \upsilon)))\mu(d\upsilon) \nonumber\\
&=P_{\mathrm{f}}(V(x)+\gamma\phi(0)W^{2}(\kappa+1, e))  \nonumber\\
&\quad +P_{\mathrm{s}}(V(x)+\gamma\phi(0)W^{2}(\kappa+1, \mathfrak{h}(\kappa+1, e)) \nonumber\\
&\leq V(x)+(\rho^{2}P_{\mathrm{s}}+\lambda^{2}P_{\mathrm{f}})\gamma\phi(0)W^{2}(\kappa, e)  \nonumber\\
&\leq U(\xi).
\end{align}
Hence, item (iii) in Theorem \ref{thm-1} holds and the proof is completed.


\end{document}